\documentclass[journal,onecolumn,draftclsnofoot]{IEEEtran}

\usepackage{amsmath,amssymb,amsfonts}
\usepackage{mathtools}
\usepackage{amsthm}
\usepackage{algorithmic}
\usepackage{graphicx}
\usepackage{pgfplots} 
\pgfplotsset{compat=1.15}% remove warning

\usepackage{pgfgantt}
\usepackage{pdflscape}
\usepackage{pst-plot}
\usetikzlibrary{spy}
\usetikzlibrary{positioning,calc}
\usetikzlibrary{%decorations.pathmorphing,
calc,shapes,shapes.geometric,patterns}
\usetikzlibrary{shapes.multipart}
\usepackage{xfrac}
\usepackage{colortbl}
\usepackage{cancel}
\usetikzlibrary{arrows,positioning,calc,intersections}
\usetikzlibrary{datavisualization.formats.functions}
\def\BibTeX{{\rm B\kern-.05em{\sc i\kern-.025em b}\kern-.08em
    T\kern-.1667em\lower.7ex\hbox{E}\kern-.125emX}}
    \newcommand{\norm}[1]{\left\lVert#1\right\rVert}
\usepgfplotslibrary{fillbetween}
\usetikzlibrary{arrows}%, decorations.markings}
\usepackage[caption=false,font=normalsize,labelfont=sf,textfont=sf]{subfig}
\usepackage{amssymb,bm}
\usepackage{float}
\usepackage{graphicx}%\usepackage[pdftex]{graphicx}
\usepackage{amsmath}%\usepackage[cmex10]{amsmath}
\usepackage{hyperref}
\usepackage{multirow}
\usepackage{xcolor}
\usepackage{tikz}
\usepackage{mathrsfs}
\usepackage{tikz}
\usepackage[american,siunitx]{circuitikz}
\usetikzlibrary{arrows,calc,positioning}
\usepackage{bbm}

\usepackage{cite} 

\usepackage{comment}

\renewcommand{\epsilon}{\varepsilon}
\newcommand{\setd}{\ensuremath{\mathcal{D}}}

\tikzset{ar/.style={-latex,shorten >=-1pt, shorten <=-1pt}}

\newcommand\underrel[3][]{\mathrel{\mathop{#3}\limits_{%
      \ifx c#1\relax\mathclap{#2}\else#2\fi}}}

\def \A{\mathcal{A}}

\def \D{\mathcal{D}}
\def \E{\mathcal{E}}
\def \F{\mathcal{F}}
\def \G{\mathcal{G}}

\def \M{\mathcal{M}}
\def \N{\mathcal{N}}

\def \P{\mathcal{P}}

\def \S{\mathcal{S}}
\def \T{\mathcal{T}}
\def \U{\mathcal{U}}

\def \W{\mathcal{W}}
\def \X{\mathcal{X}}
\def \Y{\mathcal{Y}}
\def \Z{\mathcal{Z}}

\def \sG {\mathscr{G}}

\def \fu{\mathbf{u}}

\def \fx{\mathbf{x}}
\def \fy{\mathbf{y}}

\def \fY{\mathbf{Y}}
\def \fZ{\mathbf{Z}}
\def \f0{\mathbf{0}}

\definecolor{blau_1a}{RGB}{93,133,195}
\definecolor{blau_2a}{RGB}{0,156,218}
\definecolor{gruen_3a}{RGB}{80,182,149}
\definecolor{gruen_4a}{RGB}{175,204,80}
\definecolor{gruen_5a}{RGB}{221,223,72}
\definecolor{orange_6a}{RGB}{255,224,92}
\definecolor{orange_7a}{RGB}{248,186,60}
\definecolor{rot_8a}{RGB}{238,122,52}
\definecolor{rot_9a}{RGB}{233,80,62}
\definecolor{lila_10a}{RGB}{201,48,142}
\definecolor{lila_11a}{RGB}{128,69,151}

\definecolor{blau_1b}{RGB}{0,90,169}
\definecolor{blau_2b}{RGB}{0,131,204}
\definecolor{gruen_3b}{RGB}{0,157,129}
\definecolor{gruen_4b}{RGB}{153,192,0}
\definecolor{gruen_5b}{RGB}{201,212,0}
\definecolor{orange_6b}{RGB}{253,202,0}
\definecolor{orange_7b}{RGB}{245,163,0}
\definecolor{rot_8b}{RGB}{236,101,0}
\definecolor{rot_9b}{RGB}{230,0,26}
\definecolor{lila_10b}{RGB}{166,0,132}
\definecolor{lila_11b}{RGB}{114,16,133}

\definecolor{mycolor1}{RGB}{249,245,233}

\usepackage{accents}
\newcommand{\hX}{\widehat{X}}
\newcommand{\hY}{\widehat{Y}}
\theoremstyle{remark}	\newtheorem{theorem}{Theorem}
\theoremstyle{remark}	\newtheorem{lemma}[theorem]{Lemma}
\theoremstyle{remark}	\newtheorem{corollary}[theorem]{Corollary}
\theoremstyle{remark}	
\theoremstyle{remark} \newtheorem{definition}{Definition}
\theoremstyle{remark} \newtheorem{remark}{Remark}
\theoremstyle{remark} \newtheorem{example}{Example}

\usepackage{tikz}
\usetikzlibrary{arrows}

\usepackage{circuitikz}
\usetikzlibrary{positioning}
\usepackage{circuitikz}
\usetikzlibrary{arrows.meta,positioning}

\usepackage{comment}

\usepackage{autobreak}
\allowdisplaybreaks

\begin{document}
%%%%%%%%%%%%%%%%%%%%%%%%%%%%%%%%%%%%%%%%%%%%%%%%%%%%%%%%%%%%%%%%%%%%%
\title{Deterministic Identification Over Channels With Power Constraints}

\author{
	\vspace{0.1cm}
    \IEEEauthorblockN{Mohammad J. Salariseddigh\IEEEauthorrefmark{1}, Uzi Pereg\IEEEauthorrefmark{1}, 
    Holger Boche\IEEEauthorrefmark{2}, and Christian Deppe\IEEEauthorrefmark{1}} \\
	\vspace{0.25cm}
    \IEEEauthorblockA{\normalsize \IEEEauthorrefmark{1}Institute for Communications Engineering, Technical University of Munich \\
    \IEEEauthorrefmark{2}Chair of Theoretical Information Technology, Technical University of Munich\\
    {\tt Email}: {\{mohammad.j.salariseddigh, uzi.pereg, boche, christian.deppe\}@tum.de}}
}
%%%%%%%%%%%%%%%%%%%%%%%%%%%%%%%%%%%%%%%%%%%%%%%%%%%%%%%%%%%%%%%%%%%%%
\maketitle
%%%%%%%%%%%%%%%%%%%%%%%%%%%%%%%%%%%%%%%%%%%%%%%%%%%%%%%%%%%%%%%%%%%%%
\begin{abstract}
The identification capacity is developed without randomization at neither the encoder nor the decoder.
In particular, full characterization is established for the deterministic identification (DI) capacity for the Gaussian channel and for the general discrete memoryless channel (DMC) with and without constraints. 
 Originally, Ahlswede and Dueck established the identification capacity with local randomness given at the encoder, resulting in a double exponential number of messages in the block length. In the deterministic setup, the number of messages scales exponentially, as in Shannon's transmission paradigm, but the achievable identification rates can be significantly higher than those of the transmission rates. 
 Ahlswede and Dueck further stated a capacity result for the deterministic setting of a DMC, but did not provide an explicit proof.  In this paper, a detailed proof is given for both the Gaussian channel and the general DMC.
 The DI capacity of a Gaussian channel is infinite regardless of the noise.
\end{abstract}
%%%%%%%%%%%%%%%%%%%%%%%%%%%%%%%%%%%%%%%%%%%%%%%%%%%%%%%%%%%%%%%%%%%%%
%%%%%%%%%%%%%%%%%%%%%%%%%%%%%%%%%%%%%%%%%%%%%%%%%%%%%%%%%%%%%%%%%%%%%
\IEEEpeerreviewmaketitle
%%%%%%%%%%%%%%%%%%%%%%%%%%%%%%%%%%%%%%%%%%%%%%%%%%%%%%%%%%%%%%%%%%%%%
%%%%%%%%%%%%%%%%%%%%%%%%%%%%%%%%%%%%%%%%%%%%%%%%%%%%%%%%%%%%%%%%%%%%%
%%%%%%%%%%%%%%%%%%%%%%%%%%%%%%%%%%%%%%%%%%%%%%%%%%%%%%%%%%%%%%%%%%%%%
\section{Introduction}

In the fundamental communication paradigm considered by
\IEEEPARstart{S}{hannon} \cite{S48},  a sender wishes to convey a message through a noisy channel in a such a manner that  the receiver will be able to retrieve the original message. In other words, the decoder's task is to determine which message was sent.  Ahlswede and Dueck \cite{AD89} introduced a scenario of a different nature where the decoder only performs identification and determines whether a particular message was sent or not \cite{AD89,HanBook,SCR20-2}. Applications include identification plus transmission (point-to-multipoint communication) \cite{HanVerdu}, communication complexity \cite{T99}, private interrogation theory \cite{B09}, the tactile internet \cite{TactileInternet14}, vehicle-to-X communications \cite{KBMRAZT16,JVNRCR16}, digital watermarking \cite{M01,S01,AN06}, online sales \cite{GCE07,GC08}, industry 4.0 \cite{Industry4.0-1,Industry4.0-2,Industry4.0-3}, health care \cite{Healthcare04-Patent}, molecular communications \cite{Bush15,Hasel19,6G+} and other event-triggered systems.

We present two motivating examples for applications of the identification scheme. Molecular communication is a promising candidate for the sixth generation of cellular communication networks (6G) \cite{6G+,6G_PST}, in which some applications demand alerts to be identified \cite{6G_PST}. Furthermore, in other systems of molecular communication, a nano-device needs to determine the occurrence of a specific event. For instance, in the targeted drug delivery \cite{Muller04,Nakano13} or cancer treatment \cite{Hobbs_ea98,Jain99,Wilhelm16}, a nano-device will seek to know whether the blood pH exceeds a critical threshold or not, whether a specific drug is released or not, whether another nano-device has replicated itself, whether a certain molecule was detected, whether a target location in the vessels is identified, or whether the molecular storage is empty, etc \cite{Nakano14}. A second application for identification is  vehicle-to-X communications, where a vehicle that collects sensor data may ask whether a certain alert message concerning the future movement of an adjacent vehicle was transmitted or not \cite[Sec.~VII]{BD17}.

% In molecular communications (MC) \cite{NMWVS12,FYECG16}, information is transmitted via chemical signals or molecules. In various environments, e.g. inside the human body, conventional wireless communication with electromagnetic (EM) waves is not feasible or could be detrimental. The research on micro-scale MC for medical applications, such as intra-body networks, is still in its early stages and faces many challenges.
% MC is a promising contender for future applications such as 7G+.

The identification problem \cite{AD89} can be regarded as a \emph{Post-Shannon} \cite{CGC03} model where the decoder does not perform an estimation, but rather a binary hypothesis test to decide between the hypotheses `sent' or `not sent', based on the observation of the channel output.
As the sender has no knowledge of the desired message that the receiver is interested in, the identification problem can be regarded as a test of many hypotheses occurring simultaneously.
The scenario where the receiver misses and does not identify his message is called a type I error, or `missed identification', whereas the event where the receiver accepts a false message is called a type II error, or `false identification'.

Ahlswede and Dueck \cite{AD89} required randomized coding for their identification-coding scheme. This means that a randomized source is available to the sender. The sender can make his encoding dependent on the output of this source. It is known that this resource cannot be used to increase the transmission capacity of discrete memoryless channels \cite{A78}.
A remarkable result of identification theory is that given local randomness at the encoder, reliable identification can be attained such that the code size, i.e., the number of messages, grows double exponentially in the block length $n$, i.e., $\sim 2^{ 2^{nR}}$ \cite{AD89}. This differs sharply from the traditional transmission setting where the code size scales only exponentially, i.e., $\sim{2^{nR}}$.
%Identification for the  Gaussian channel is considered in \cite{MasterThesis,LDB20}.
Beyond the exponential gain in identification, the extension of the problem to more complex scenarios reveals that the identification capacity has a very different behavior compared to the transmission capacity \cite{feedback,correlation,BV00,BD18_2,W04,BL17}.
For instance, feedback \emph{can} increase the identification capacity \cite{feedback} of a memoryless channel, as opposed to the transmission capacity \cite{S56}.
Nevertheless, it is difficult to implement randomized-encoder identification (RI) codes that will achieve such performance, because it requires the encoder to process a bit string of exponential length. The construction of identification codes is considered in \cite{SCR20-2,VK93,KT99,Bringer09,Bringer10}. Identification for Gaussian channels is considered in \cite{MasterThesis,LDB20,Labidi2021,Ezzine2021,BV00}.

In the deterministic setup, the number of messages scales exponentially in the blocklength \cite{AD89,AN99,J85,Bur00}, as in the traditional setting of transmission. Nevertheless, the achievable identification rates are significantly higher than those of transmission. In addition, deterministic codes often have the advantage of simpler implementation and simulation \cite{PP09}, explicit construction \cite{A09}, and single-block reliable performance. In particular, J\'aJ\'a \cite{J85} showed that the deterministic identification (DI) capacity\begin{footnote}{The DI capacity in the literature is also referred to as the non-randomized identification (NRI) capacity \cite{AN99} or the dID capacity \cite{BV00}.}\end{footnote}\;of a binary symmetric channel is 1 bit per channel use, as one can exhaust the entire input space and assign (almost) all sequences in the $n$-dimensional space $\{0,1\}^n$ as codewords. Ahlswede et al. \cite{AD89,AN99} stated that the DI capacity of a discrete memoryless channel (DMC) with a stochastic matrix $W$ is given by the 
logarithm of the number of distinct row vectors of $W$ (\cite[see Sec.~IV]{AD89} and \cite[see Abstr.]{AN99}).
Nonetheless, an explicit proof for this result was not provided in \cite{AD89,AN99}. 
Instead, Ahlswede and Cai \cite{AN99} referred the reader to a paper \cite{A80} which does not include identification and addresses a completely different model of an arbitrarily varying channel \cite{A80}. Since then, the problem of proving this result has remained unsolved, since a straightforward extension of the methods in \cite{A80}, using decoding territories, does not seem to yield the desired result on the DI capacity \cite{CaiEmail2}.

In this paper, we establish the DI capacity of a channel subject to an input constraint. Such a constraint is often associated with a limited power supply or regulation, as in the case of the Gaussian channel. Our main result is that the DI capacity of a DMC $\W$, under the input constraint $\frac{1}{n}\sum_{t=1}^n \phi(x_t)\leq A$, is given by
%%%
\begin{align}
    \mathsf{C}_{DI}(\W)=\max_{p_X \;:\, \mathbb{E}\{\phi(X)\} \leq A} H(X) \;,\,
\end{align}
%%%
and that the DI capacity of a Gaussian Channel $\mathscr{G}$ under power constraint $A$ is infinite, regardless of the noise in the channel.
For a DMC, we may assume without loss of generality that the rows of the channel matrix are distinct (see Section~\ref{Subsec.ChannelReduction}). This result has the following geometric interpretation. At first glance, it may seem reasonable that for the purpose of identification, one codeword could represent two messages. While identification allows overlap between decoding regions \cite{KE05,AADT20}, it turns out that overlap at the encoder is not allowed for deterministic codes. We observe that if two messages are represented by the same codeword, then the low probability of a type I error comes at the expense of the high probability of a type II error, and vice versa.
That is, as shown in our proof, if the  probability of missed identification is upper bounded by $\epsilon$, then the probability of false identification is lower bounded by $1-\epsilon$.
Thus, DI coding imposes the restriction that every message must have a distinct codeword.
The converse proof follows from this property in a  straightforward manner since the volume of the input subspace of sequences that satisfy the input constraint is 
$\approx 2^{n\mathsf{C}_{DI}(\W)}$.
A similar principle guides the direct part as well. The input space is covered such that each codeword is surrounded by a sphere of radius $n\epsilon$ to separate the codewords. For the Gaussian channel, the DI capacity can be achieved using a simple distance-decoder.

 By providing a detailed proof for this problem, we thus fill the gap in the previous analysis  \cite{AD89,AN99} as well.
In the proof, we use the method of types, while the derivation is based on ideas that are analogous to the combinatoric analysis of Hamming distances by J\'aJ\'a \cite{J85}.
Although the codebook construction is similar to that of Ahlswede's coding scheme  \cite{A80}, the decoder is significantly different. In particular, we do \emph{not} use decoding territories as in \cite{A80}, but rather perform a typicality check.
Nonetheless, the type-class intersection lemma and the message-set analysis in \cite{A80} turn out to be useful in our analysis as well. %(see Lem. $I_1$ and Sec. 2 therein). 
Hence, our proof combines techniques and ideas from both works, by J\'aJ\'a \cite{J85} and by Ahlswede \cite{A80}, to derive the DI capacity both with and without an input constraint.
%%%%%%%%%%%%%%%%%%%%%%%%%%%%%%%%%%%%%%%%%%%%%%%%%%%%%%%%%%%%%%%%%%%%%%%%%%%%%%%%%%%%%%%%%%%%%%%%%%%%%%%%%%%%%%%%%%%%%%%%%%%%%%%%%%%%%%%%%%%%%%%%%%%%%%%%%%%%%%%%%%%%%%%%%%%%%%%%%%%%%%%%%%%%%%%%%%%%%%%%%%%%%%%
\section{Definitions and Related Works} \label{sec:preliminaries}
In this section, we introduce the channel model and coding definitions. Here, we only consider the discrete memoryless channel (DMC). The channel description and coding definition for the Gaussian channel will be presented in Section~\ref{Sec.GaussianChannel}.

\subsection{Notation}
 We use the following notation conventions throughout. Calligraphic letters $\X,\Y,\Z,\ldots$ are used for finite sets.
Lowercase letters $x,y,z,\ldots$  stand for constants and values of random variables, and uppercase letters $X,Y,Z,\ldots$ stand for random variables.  
 The distribution of a random variable $X$ is specified by a probability mass function (pmf) $p_X(x)$ over a finite set $\X$. The set of all pmfs over $\X$ is denoted by $\P(\X)$, $H(X)$, and $I(X ;Y)$ are the entropy and mutual information, respectively;  all logarithms and information quantities are taken to the $2$. We use $x^j=(x_1,x_{2},\ldots,x_j)$ to denote  a sequence of letters from $\X$.
 A random sequence $X^n$ and its distribution $p_{X^n}(x^n)$ are defined accordingly. 
The set of consecutive natural numbers from $1$ through $M$ is denoted by $[\![M]\!]$.
The Hamming distance between two sequences $a^n$ and $b^n$ is defined as the number of positions for which the sequences have different symbols, i.e.,
% ---
\begin{align}
    d_H (a^n,b^n) = \left| \left\{ t \in [\![n]\!] ~;~ a_t \neq b_t \right\} \right| \;.\,
\end{align}
% ---
%denotes the indicator function of the set $A$.
The $n$-dimensional Hamming sphere of radius $n\epsilon$ that is centered at $a^n$ is defined as
%%%
\begin{align}
    \S_{\epsilon}(a^n) = \left\{ x^n\in\X^n \,:\; d_H(x^n,a^n) < n\epsilon \right\} \;.\,
\end{align}
%%%
Further, we denote the hyper-sphere of radius 
$r$ around $\fx_0$ by
%%%
\begin{align} 
    \S_{\fx_0}(n,r) = \left\{ \fx \in \mathbb{R}^n \,:\; \norm{\fx-\fx_0} \leq r \right\} \;.\,
\end{align}
%%%
In the continuous case, we use the cumulative distribution function 
$F_X(x)=\Pr(X\leq x)$ for $x\in\mathbb{R}$, or alternatively, the probability density function (pdf) $f_X(x)$,  when it exists.
 The notation $\mathbf{x}=(x_1,x_{2},\ldots,x_n)$ is used instead of $x^n$  when it is understood from the context that the length of the sequence is $n$, and  the $\ell^2$-norm of $\mathbf{x}$ is denoted  by $\norm{\mathbf{x}}$. 
%%%%%%%%%%%%%%%%%%%%%%%%%%%%%%%%%%%%%%%%%%%%%%%%%%%%%%%%%%%%%%%%%%%%%
\subsection{Channel Description}
%\subsubsection*{Discrete Memoryless Channel (DMC)}
\label{Subsec.ChannelDescription}
A DMC $(\X,\Y,W)$ consists of  finite input and output alphabets $\X$ and $\Y$, respectively, 
and a  conditional pmf $W(y|x)$.
The channel is memoryless without feedback, and therefore   $W^n(y^n|x^n)= \prod_{t=1}^n W(y_t|x_t)$. We denote a DMC by 
$\W= (\X,\Y,W)$.
%%%
Next, we consider an input constraint. 
Let $\phi:\X\rightarrow [0,\infty)$ be some given bounded cost function, and define
	\begin{align}
	\label{Eq.Cost1}
	\phi^n(x^n)=&\frac{1}{n} \sum_{t=1}^n \phi(x_t) \;.\,
	\end{align}
Let $A>0$. Given an input constraint $A$  corresponding to  the cost function
$\phi^n(x^n)$, the channel input $x^n$ must satisfy
% ---
\begin{align}
    \label{Eq.Cost}
    \phi^n(x^n)\leq & A \;.\,
\end{align}
% ---
We may assume without loss of generality that $0\leq A\leq \phi_{\text{\,max}}$, where $\phi_{\text{\,max}}$ is given by
% ---
\begin{align}
    \phi_{\text{\,max}} = \underset{x\in\X}{\max}~\phi(x) \;.\,
\end{align}
% ---
It is also assumed that for some $x_0\in\X$, $\phi(x_0)=0$.
%
% %%%%%%%%%%%%%%%%%%%%%%%%%%%%%%%%%%%%%%%%%%%%%%%%%%%%%%%%%%%%%%%%%%%%%
\subsection{Coding}
The definitions for DI codes, achievable rates, and capacity are given below.
%%%
\begin{definition}
\label{deterministicIDCode}
 A  $(2^{nR},n)$ DI code for a DMC $\W$ under input constraint $A$, assuming $2^{nR}$ is an integer, is defined as a system $(\U,\mathscr{D})$
that consists of 
 a codebook $\U=\{ u_i \}_{i\in[\![2^{nR}]\!]}$, $\U\subset \X^n$,
 such that
 % ---
\begin{align} 
    \phi^n(u_i)\leq A \;,\,
\end{align}
% ---
for all $i\in[\![2^{nR}]\!]$ and a collection of decoding regions $\mathscr{D}=\{ \D_i \}_{i\in[\![2^{nR}]\!]}$ with $\bigcup_{i=1}^{2^{nR}}\D_i\subset \Y^n$. Given a message $i\in [\![2^{nR}]\!]$, the encoder transmits $u_i$. The decoder's aim is to answer the following question: Was a desired message $j$ sent or not? 
Two types of errors may occur: 
Rejection of the true message, or accepting a false message. Those error events are often referred to as type I and type II errors, respectively.
Specifically, $P_{e,1}^{(n)}(i)$ is the  type I error probability for rejecting the true message $i$, while $P_{e,2}^{(n)}(i,j)$ is the   type II error probability for accepting the false message $j$ given that the message $i$ was sent.

 The error probabilities of the identification code $(\U,\mathscr{D})$ are given by
\begin{align}
 P_{e,1}(i)&= W^n(\setd_i^c|u_i)  &&\hspace{-2cm}  \text{(missed-identification error)} \;,\, \label{Eq.TypeIErrorDef} \\
 P_{e,2}(i,j) & = W^n(\setd_j|u_i)  &&\hspace{-2cm} \text{(false identification error)} \;.\,
 \label{Eq.TypeIIErrorDef}
\end{align}
%%%
A $(2^{nR},n,\lambda_1,\lambda_2)$ DI code further satisfies
% ---
\begin{align}
    \label{Eq.TypeIError}
    P_{e,1}(i) &\leq \lambda_1 \;,\,
    \\
    \label{Eq.TypeIIError}
    P_{e,2}(i,j) &\leq \lambda_2 \;.\,
\end{align}
% ---
for all $i,j\in [\![2^{nR}]\!]$ such that
$i\neq j$.
%%%

A rate $R>0$ is called achievable if for every
$\lambda_1,\lambda_2>0$ and sufficiently large $n$, there exists a $(2^{nR},n,\lambda_1,\lambda_2)$ DI code. 
The operational DI capacity is defined as the supremum of achievable rates, and will be denoted by $\mathbb{C}_{DI}(\W)$.

\end{definition}

Alternatively, one may consider achievable identification rates for codes with a \emph{double}-exponential number of messages \cite{AD89}. A rate $R>0$ is said to be achievable in the \emph{double}-exponential scale if there exists a corresponding
 $(2^{2^{nR}},n,\lambda_1,\lambda_2)$ DI code. We denote the DI capacity in the double-exponential scale by $\underline{\mathbb{C}}_{DI}(\W)$.

As mentioned earlier, Ahlswede and Dueck \cite{AD89} needed randomized encoding for their identification-coding scheme. This means that a randomized source is available to the sender. The sender can make his encoding dependent on the output of this source.
Therefore, a randomized-encoder identification (RI) code is defined in a similar manner where the encoder is allowed to select a codeword $U_i$ at random according to some conditional input distribution $Q(x^n|i)$.
The RI capacities in the exponential and double-exponential scales are then  denoted by $\mathbb{C}_{RI}(\W)$ and $\underline{\mathbb{C}}_{RI}(\W)$, respectively.
Given local randomness at the encoder, reliable identification can be attained such that 
the number of messages grow double exponentially in the block length $n$, i.e., $\sim 2^{ 2^{nR}}$ \cite{AD89}. This differs sharply from the traditional transmission setting where the code size scales only exponentially, i.e., $\sim{2^{nR}}$. Remarkably, in \cite{AD89} it is shown that $\underline{\mathbb{C}}_{RI}(\W)=\mathbb{C}_T(\W)$, 
where $\mathbb{C}_T(\W)$ denotes the transmission capacity of the channel.

\begin{remark}
Observe that in general, if the capacity in an exponential scale is finite, then it is zero in the double exponential scale. Conversely, if the capacity in a double exponential scale is positive, then the capacity in the exponential scale is $+\infty$.
\end{remark}

\begin{remark}
The MC has recently made advances on the technological side. This development is about promoting complex networks, such as the Internet of Things (IoT), with MC. 
The IoT describes the integration of intelligent/smart machines and objects on the Internet. 
These smart devices can be accessed and controlled via the Internet. 
The advances made in the field of nanotechnology enable
the development of devices in the nano-meter range, which are called nanothings. 
The interconnection of nanothings with the Internet is known as Internet of NanoThings (IoNT)
and is the basis for various future healthcare and military applications \cite{Dress15}. 
Nanothings are based on synthesized materials, use electronic circuits, and EM-based communication. Unfortunately, these characteristics could be harmful for some application environments, such
as inside the human body. The concept of Internet of Bio-NanoThings (IoBNT) has been introduced in \cite{Aky15}, where nanothings are biological cells that are created using tools from synthetic biology and nanotechnology. Such biological nanothings are called bio-nanothings.
Similar to artificial nanothings, bio-nanothings have control (cell nucleus), power (mitochondrion), communication (signal pathways), and sensing/actuation (flagella, pili or cilia) units. For the communication between cells, MC is especially well suited, since the natural exchange of information between cells is already based on this paradigm. MC in cells is based on signal pathways (chains of chemical reactions) that process information that is modulated into chemical characteristics,
such as molecule concentration. Identification is a very interesting communication task for these applications. However, it is also unclear how RI codes can be incorporated into such systems. It is unclear how powerful random number generators should be developed for synthetic materials on these small scales. In the case of Bio-NanoThings, it is uncertain whether the natural biological processes can be controlled or reinforced by local randomness at this level. Therefore, for the design of synthetic IoNT, or for the analysis and utilization of IoBNT, it is interesting to consider identification with deterministic encoding.
\end{remark}

A geometric illustration for the type I and II error probabilities is given in Figure~\ref{Fig.GeometricID}. When the encoder sends the message $i$, but the channel output is outside $\D_i$, then type a I error occurs. This kind of error is also considered in traditional transmission. In identification, the decoding sets can overlap. A type II error covers the case where the output sequence belongs to the intersection of $\D_i$ and $\D_j$ for $j\neq i$.
%%%
\begin{figure}[ht!]
    \label{Fig.GeometricID}
    \scalebox{1}{
\begin{tikzpicture}[scale=.7][thick]
%%%%%%%%%%%%%%%%%%%%%%%%%%%%%%%%%%%%%%%%%%%%%%%%%%%%%%%%%%%%%%%%%%%%%%%%%%%%%%%%%%%%%%%%%%%%%%%%
\draw[dashed] (0,.7) circle (3.6cm);
\draw[->] (0,4.4) -- ++(0,1.5)  node [fill=white,inner sep=6pt](a){$\X^n$};

\node[circle,draw, dash dot, thick, minimum size=.25mm] (u1) at (-2,2) {};

% \node[circle,draw, minimum size=.25mm] (u2) at (-1.8,-.8) {};

% \node[circle,draw, minimum size=.25mm] (u3) at (1,3) {};

% \node[circle,draw, minimum size=.25mm] (u4) at (1.4,-1.8) {};

\node at (-2,1.5) (u1n) {$u_1$};
\node at (-1.8,-1.2) (u2n) {$u_2$};
\node at (1,2.4) (u3n) {$u_3$};
\node at (1.4,-1.8) (u4n) {$u_4$};

\draw [fill=gray!30!white, fill opacity=0.5, name path=i] (1,2.8) circle (.1cm);
\draw [fill=cyan, fill opacity=0.5, name path=k] (-2,2) circle (.1cm);
\draw [fill=red, fill opacity=0.3, name path=k] (1.4,-1.4) circle (.1cm);
\draw [fill=yellow, fill opacity=0.4, name path=k] (-1.8,-.8) circle (.1cm);
%%%%%%%%%%%%%%%%%%%%%%%%%%%%%%%%%%%%%%%%%%%%%%%%%%%%%%%%%%%%%%%%%%%%%%%%%%%%%%%%%%%%%%%%%%%%%%%%
\draw[dashed] (14.3,.7) circle (4.4cm);
\draw[->] (14.3,5.2) -- ++(0,1.5)  node [fill=white,inner sep=6pt](a){$\Y^n$};

\draw (13,2) circle (1.5cm);
\draw (14,0) circle (1.5cm);
\draw (15,2.5) circle (1.25cm);
\draw (16,1) circle (1.4cm);

\node at (13,2) (D1n) {$\D_1$};
\node at (14,0) (D2n) {$\D_2$};
\node at (15,2.5) (D3n) {$\D_3$};
\node at (16,1) (D4n) {$\D_4$};

\draw [fill=cyan, fill opacity=0.5, name path=k] (13,2) circle (1.5cm);

\draw [fill=yellow, fill opacity=0.5, name path=k] (14,0) circle (1.5cm);

\draw [fill=gray, fill opacity=0.1, name path=i] (15,2.5) circle (1.25cm);

\draw [fill=red, fill opacity=0.4, name path=j] (16,1) circle (1.4cm);

%%%%%%%%%%%%%%%%%%%%%%%%%%%%%%%%%%%%%%%%%%%%%%%%%%%%%%%%%%%%%%%%%%%%%%%%%%%%%%%%%%%%%%%%%%%%%%%%

\path (u1) edge [-> , thick, blue, bend left] node [sloped,midway,above,font=\small] {correct identification}(D1n);

\path (u1) edge [-> , thick, red, bend right] node [sloped,midway,below,font=\small] {type I error}(16.5,0.2);

\path (u1) edge [-> , thick, brown!80!black, bend right] node [sloped,midway,above,font=\small] {type II error}(13.3,1);

%%%%%%%%%%%%%%%%%%%%%%%%%%%%%%%%%%%%%%%%%%%%%%%%%%%%%%%%%%%%%%%%%%%%%%%%%%%%%%%%%%%%%%%%%%%%%%%%
\end{tikzpicture}
}
   \caption{Geometric illustration of  identification errors in the deterministic setting. The arrows indicate three scenarios for the channel output, given that the encoder transmitted the codeword $u_1$ corresponding to $i=1$. If the channel output is outside $\D_1$, then a type I error has occurred, as indicated by the  bottom red arrow. This kind of error is also considered in traditional transmission. In identification, the decoding sets can overlap. If the channel output belongs to $\D_1$ but also belongs to $\D_2$, then a type II error has occurred, as indicated by the middle brown arrow. Correct identification occurs when the channel output belongs \emph{only} in $\D_1$, which is marked in blue.}
\end{figure}
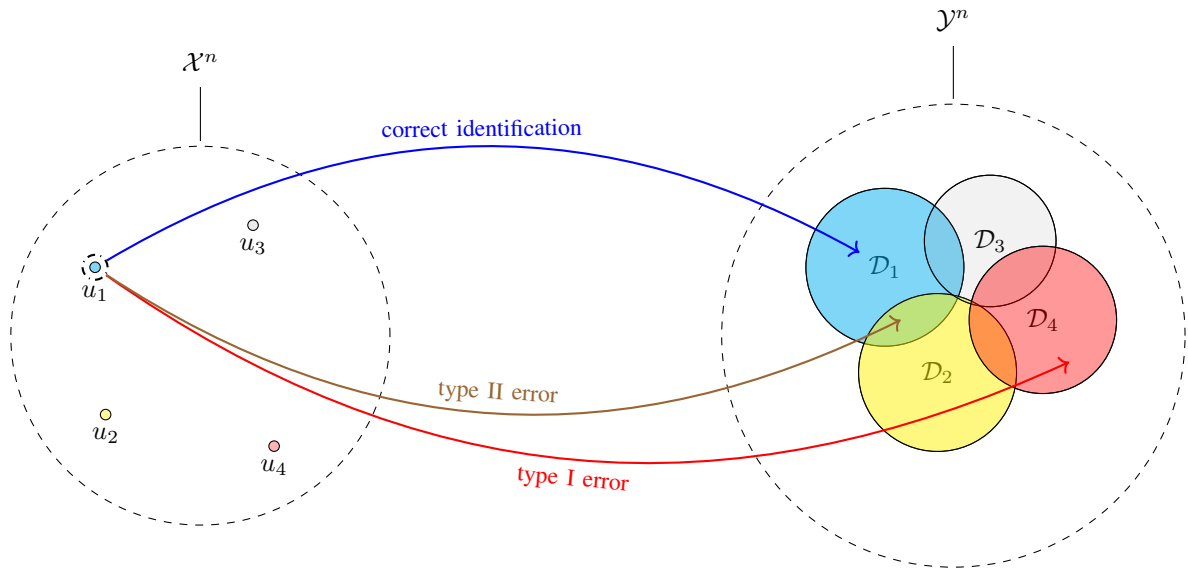
%%%
%%%%%%%%%%%%%%%%%%%%%%%%%%%%%%%%%%%%%%%%%%%%%%%%%%%%%%%%%%%%%%%%%%%%%%
\subsection{Related Work}
%\subsubsection{Local Randomization At Encoder}
We briefly review Ahlswede and Dueck's result \cite{AD89} on the RI capacity, i.e., when the encoder uses a stochastic mapping. As mentioned above, using RI codes, it is possible to identify a double-exponential number of messages in the block length $n$.
That is, given a rate $R<\underline{\mathbb{C}}_{RI}(\W)$, there exists a sequence of $(2^{2^{nR}},n)$ RI codes with vanishing error probabilities.
Despite the significant difference between the definitions in  the identification setting and in the transmission setting, it was shown that the value of the RI capacity in the double-exponential scale equals the Shannon capacity of transmission.
%%%
\begin{theorem}[see \cite{AD89, HanVerdu}]
\label{Th.DMCIdentificationCapacity}
	The RI capacity in the double-exponential scale of a DMC $\W$ is given by
	%%%
	\begin{align}
	\underline{\mathbb{C}}_{RI}(\W) =\max_{p_X \,:\; \mathbb{E}\{ \phi(X)\}\leq A} I(X;Y) \;.\,
	\end{align} \label{IDtheorem}
	%%%
	Hence, the RI capacity in the exponential scale is infinite, i.e.,
	\begin{align}
	\mathbb{C}_{RI}(\W)=\infty \;.\,
	\end{align}
\end{theorem}
%%%
In the next sections, we will consider the identification setting when the encoder does not have access to randomization.
%%%
%%%%%%%%%%%%%%%%%%%%%%%%%%%%%%%%%%%%%%%%%%%%%%%%%%%%%%%%%%%%%%%%%%%%%%%%%%%%%%%%%%%%%%%%%%%%%%%%%%%%%%%%%%%%%%%%%%%%%%%%%%%%%%%%%%%%%%%%%%%%%%%%%%%%%%%%%%%%%%%%%%%%%%%%%%%%%%%%%
\section{Main Results - DMC}
We give our main results on the DI capacity of the DMC.
For a DI code, as opposed to the randomized case, the number of messages $2^{nR}$ is only exponential in the blocklength. In this sense, DI codes are similar to transmission codes.  
However, the achievable rates for identification are significantly higher, as the DI capacity is given in terms of the input entropy instead of the mutual information.

\subsection{Channel Reduction}
\label{Subsec.ChannelReduction}
We begin with  a procedure of channel reduction where we remove identical rows from the channel matrix,
so that the remaining input letters have a lower cost compared to the deleted letters. As will be seen below, the DI capacity remains the same following this reduction.
The characterization of the DI capacity will be given in the next section in terms of the reduced input alphabet.

We begin with the definition of the reduced channel.
\begin{definition}[Reduced channel]
\label{Def.ReducedChannel}
Given a DMC $\W$  with a stochastic matrix $W:\X\to\Y$, we define the reduced DMC $W_r$ as follows.
Let $ \{\X(\ell)\}$ be a partition of $\X$ into equivalent classes, so that two letters $x$ and $x'$ belong to the same equivalent class if and only if the corresponding rows are identical, namely
% ---
\begin{align}
    x,x'\in\X(\ell) \;\Leftrightarrow\;
    W(y|x)=W(y|x') \; \quad \forall y\in\Y \;.\,
\end{align}
% ---
For every class $\X(\ell)$, assign a representative element
% ---
\begin{align}
z(\ell) = \arg\min_{x\in\X(\ell)}~\phi(x) \;,\,
\end{align}
% ---
which is associated with the lowest input cost. If there is more than one letter that is associated with the lowest input cost in $\X(\ell)$, then choose one of them arbitrarily.
Then the reduced input alphabet is defined as
%%%
\begin{align}
    \label{Eq.X_r}
    \X_r=\{ z(\ell) \} \;,\,
\end{align}
%%%
and the reduced DMC $\W_r$ is defined by a channel matrix $W_r:\X_r\to\Y$, consisting of the rows in $\X_r$, i.e.,
%%%
\begin{align}
    W_r(y|x)=W(y|x) \;,\,
\end{align}
for $x\in\X_r$ and $y\in\Y$.
\end{definition}
\begin{remark}
%%%
We note that the reduction procedure above can be viewed as merging input letters with identical rows in $W$. Furthermore, the channels $\W$ and $\W_r$ are \emph{equivalent} in the sense that $W$ and $\W_r$ are degraded with respect to each other \cite[Section III]{IA13}.
Thus, the property in the lemma below is not surprising.
\end{remark}
%%%%%%%%%%%%%%%%%%%%%%%%%%%%%%%%%%%%%%%%%%%%%%%%%%%%%%%%%%%
%%%
\begin{lemma}
\label{Lem.Reduction}
 The operational capacities of the reduced channel $\W_r$ and the original channel $\W$ are the same:
%%%
\begin{align}
    \mathbb{C}_{DI}(\W) = \mathbb{C}_{DI}(\W_r) \;.\,
\end{align}
%%%
\end{lemma}
%%%
We give the proof of Lemma~\ref{Lem.Reduction} in Appendix~\ref{App.Reduction}.
%%%
As we will see shortly, the DI capacity of a DMC $\W$ depends on $W$ only through $\X_r$. That is,
the DI capacity does not depend on the individual values of the channel matrix and  depends solely on the distinctness of its rows.
%%%%%%%%%%%%%%%%%%%%%%%%%%%%%%%%%%%%%%%%%%%%%%%%%%%%%%%%%%%
\subsection{Capacity Theorem}
In this section, we give our main result on the DI capacity of a channel subject to input constraint. The capacity result is stated in terms of the reduced channel as defined in the previous section. Let $\W$ be a DMC channel with input cost function $\phi(x)$ and input constraint $A$ as specified in (\ref{Eq.Cost}).
%%%%%%%%%%%%%%%%%%%%%%%%%%%%%%%%%%%%%%%%%%%
Define
\begin{align}
    \mathsf{C}_{DI}(\W) = \underset{p_X \,:\; \mathbb{E}\{\phi(X)\} \leq A}{\max}~H(X) \;,\,
\end{align}
for $X\sim p_X$.

\begin{theorem}
\label{Th.DDICapacity}
The DI capacity of a DMC   $\W$ under input constraint is given by
%%%
\begin{align}
    \mathbb{C}_{DI}(\W) = \mathsf{C}_{DI}(\W_r) \;,\,
\end{align}
%%%
where $W_r$ denotes the reduced channel (see Definition~\ref{Def.ReducedChannel}). Hence, the DI capacity in the double exponential scale is zero, i.e.,
$\underline{\mathbb{C}}_{DI}(\W)=0$.
\end{theorem}
%%%
We prove the direct part in Subsection~\ref{Subsec.Achievability} and the converse part in Subsection~\ref{Subsec.Converse}.
%%%%%%%%%%%%%%%%%%%%%%%%%%%%%%%%%%%%%%%
Notice that we have characterized the DI capacity of the DMC $\W$ in terms of its reduced version, as specified in Lemma~\ref{Lem.Reduction}

\begin{corollary}[also in {\cite{AD89,AN99}}]
\label{Co.DDICapacity0}
The DI capacity of a DMC $\W$ without  constraints, i.e., with $A=\phi_{\max}$, is given by
%%%
\begin{align}
    \mathbb{C}_{DI}(\W) = \log \left( n_{row}(W) \right) \;,\,
\end{align}
%%%
where $n_{row}(W)$ is the number of distinct rows of $W$.
\end{corollary}
The corollary above is an immediate consequence of Theorem~\ref{Th.DDICapacity}. Indeed, for $A=\phi_{\max}$, we have
%%%
\begin{align}
    \mathsf{C}_{DI}(\W_r) & = \underset{p_X \,,\; \text{supp}\{ p_X \}\subseteq \X_r}{\max}~H(X)
    \nonumber\\ 
    & = \log \left| \X_r \right|
    \nonumber\\
    & = \log \left( n_{row}(W) \right) \;,\,
\end{align}
%%%
since the maximal value of $H(X)$ is $\log|\X|$, and follows the definition of $\X_r$ in Definition~\ref{Def.ReducedChannel}.
%%%
\begin{remark}
    Ahlswede et al. \cite{AD89,AN99} stated that the result in Corollary~\ref{Co.DDICapacity0} on the DI capacity of a DMC without constraints ( \cite[see Sec.~IV]{AD89} and \cite[see Abstr.]{AN99}). Nonetheless, an explicit proof for this result was not provided in \cite{AD89,AN99}. Instead, Ahlswede and Cai \cite{AN99} referred the reader to a paper \cite{A80} which does not include identification but rather the arbitrarily varying channel \cite{A80}. Since then, the problem of proving this result has remained unsolved, as a straightforward extension of the methods in \cite{A80}, using decoding territories, does not seem to yield the desired result on the DI capacity.
\end{remark}
%%%
\begin{remark}
\label{Rem.Geometry}
Our result in Theorem~\ref{Th.DDICapacity} has the following geometric interpretation. At a first glance, it may seem reasonable that for the purpose of identification, one codeword could represent two messages. However, as can be seen in the converse proof in Subsection~\ref{Subsec.Achievability}, the deterministic setting imposes the restriction that every message must have a distinct codeword. While identification allows overlap between decoding regions \cite{KE05}, it turns out that overlap at the encoder is not allowed for deterministic codes.
The converse proof follows from this property in a  straightforward manner since the volume of the input subspace of sequences that satisfy the input constraint is 
$\approx 2^{n\mathsf{C}_{DI}(\W_r)}$. A similar principle guides the direct part as well. Namely, the input space is covered such that each codeword is surrounded by a sphere of radius $\frac{n\epsilon}{2}$ to separate the codewords. 
\end{remark}

To illustrate our results, we give the following example.
\begin{example}
 Consider the binary symmetric channel (BSC),
 % ---
 \begin{align}
     Y=X+Z \mod 2 \;,\,
 \end{align}
 % ---
 where $\X=\Y=\{0,1\}$, $Z\sim\text{Bernoulli}(\epsilon)$,
 with crossover probability $0\leq\epsilon\leq \frac{1}{2}$. Suppose that the channel is subject to a Hamming weight input constraint, 
% ---
\begin{align}
	\frac{1}{n} \sum_{t=1}^n x_t \leq A \;,\,
\end{align}
% ---
with $\phi(x)=x$. Observe that for $\epsilon=\frac{1}{2}$, the rows of the channel matrix are identical. Hence, the reduced input alphabet consists of one letter, and the DI capacity is zero (see Definition~\ref{Def.ReducedChannel}).

Now, suppose that $\epsilon<\frac{1}{2}$. Then the rows of the channel matrix $W=\left(\begin{matrix} 1-\epsilon & \epsilon \\ \epsilon& 1-\epsilon \end{matrix}\right) $ are distinct, hence $\W_r=\W$. Since the channel input is binary, 
%%%
\begin{align}
    \mathsf{C}_{DI}(\W) %= \max_{p_X \,:\; \mathbb{E}\{X\}\leq A} H(X) 
    = \max_{0\leq p\leq A} H_2(p) \;,\,
    \label{eq:BSCcapacity}
\end{align}
%%%
where $H_2(p)$ is the binary entropy function and is given by
% ---
\begin{align}
    H_2(p) = - (1-p) \log_2(1-p) - p\log_2(p) \;.\,
\end{align}
% ---
Therefore, by 
Theorem~\ref{Th.DDICapacity}, the DI capacity of the BSC with Hamming weight constraint is given by
%%%
\begin{align}
    \mathbb{C}_{DI}(\W)&=\begin{cases}
    H_2(A) &\text{if $A<\frac{1}{2}$} \;,\, \\
    1      &\text{if $A\geq \frac{1}{2}$} \;,\,
    \end{cases}
     \label{eq:BSCcapacityOp}
\end{align}
%%%
(see Figure~\ref{Fig.DICapacityvsIC}). 
To show the direct part, set $X\sim\text{Bernoulli}(A)$ if $A<\frac{1}{2}$, and $X\sim\text{Bernoulli}\big(\frac{1}{2}\big)$ otherwise. 
The converse part follows from (\ref{eq:BSCcapacity}), as the binary entropy function $H_2(p)$ is strictly increasing on $0\leq p\leq \frac{1}{2}$, attaining its maximum value $H_2(\frac{1}{2})=1$, and strictly decreasing on $ \frac{1}{2}<p\leq 1$ (see Figure~\ref{Fig.DICapacityvsIC}). In accordance with Remark~\ref{Rem.Geometry}, the geometric interpretation is that the binary Hamming ball of radius $np$ can be covered with codewords. As the volume of the Hamming ball is approximately $ 2^{nH_2(p)}$, one can achieve rates that are arbitrarily close to $H_2(p)$.
% \textcolor{red}{Please add another figure, in 3-D, that demonstrates a Hamming ball.}
Without an input constraint, i.e., for $A=1$, we recover the result of J\'aJ\'a \cite{J85},
%%%
\begin{align}
    \mathbb{C}_{DI}(\W) = 1 \;.\,
\end{align}
%%%
This example demonstrates that the DI capacity is discontinuous in the channel statistics, as $\mathbb{C}_{DI}(\W,L) = 1$ for $\epsilon < \frac{1}{2}$ and $\mathbb{C}_{DI}(\W,L) = 0$ for $\epsilon = \frac{1}{2}$.
% ---
\begin{figure}[htb]
\label{Fig.DICapacityvsIC}
	\centering
    \begin{tikzpicture}[
  declare function={
    func(\x)= and(\x >= 0, \x <= .5) * (-\x*log2(\x)-(1-\x)*log2(1-\x)) + (\x >= .5) * (1);
  }]
\begin{axis}[
            xtick={0,.5,1},xticklabels={$0$,$\frac{1}{2}$,$1$},
            ytick={0,1},
            xlabel= \text{Input Constraint, A} ,
			ylabel= DI Capacity,
			xmin = 0,
			xmax = 1.1,
			ymin = 0,
			ymax = 1.22,
			width= 80mm,
			height = 60mm,
			%legend rows=1,
			axis x line = bottom,
			axis y line = left,
            legend style={%fill=orange!15!,%anchor=north,
            at={(1.2,1.2)}},
            domain=0:1,samples=1000,
]
\addplot [rot_8b, thick,dashed] {-\x*log2(\x)-(1-\x)*log2(1-\x)};
\addlegendentry{$H_2(A)$}

\addplot [blau_2b, thick] {func(x)};
% \legend{,$H_2(A)$,$\mathbb{C}_{DI}(\mathcal{W})$}
\addlegendentry{$\mathbb{C}_{DI}(\mathcal{W})$}
\end{axis}
\end{tikzpicture} 
\caption{The deterministic identification (DI) capacity of the BSC as a function of the input constraint $A$. The dashed red line indicates the binary entropy function, which is maximized in (\ref{eq:BSCcapacity}). The solid blue line indicates the DI capacity.}
\end{figure}
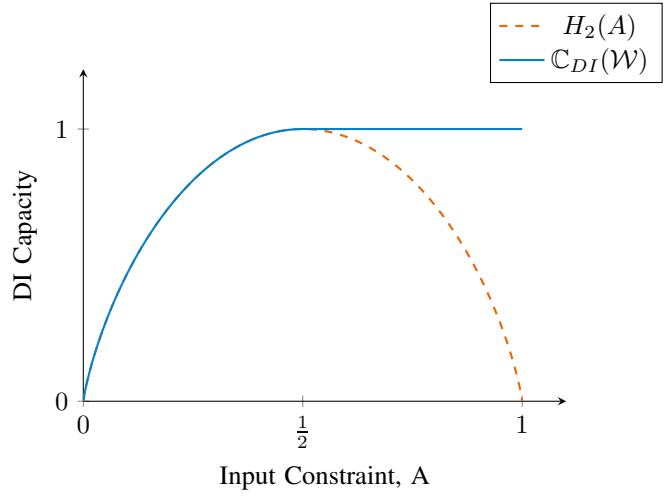 
% ---
\end{example}
% ------------------------------
\subsection{Achievability proof}
\label{Subsec.Achievability}
Consider a DMC $\W$. By Lemma~\ref{Lem.Reduction} we can assume without loss of generality that the channel matrix $W:\X\rightarrow \Y$ has distinct row vectors.
To prove achievability of the DI capacity, we combine methods and ideas from the work of J\'aJ\'a \cite{J85} as well as techniques by Ahlswede \cite{A80}. The analysis for the type II error is based on ideas that are analogous to the combinatoric analysis of Hamming distances in \cite{J85}.
The codebook construction is similar to that of Ahlswede's coding scheme  \cite{A80}, % namely random codebook with Hamming distance property, 
yet the decoder is significantly different. % and is not based on decoding territories as in \cite{A80}.
Nonetheless, the type-class intersection lemma and the message-set analysis in \cite{A80} are useful in our analysis for the type II error.
 %First, we show that there exists a code  such that the codewords are separated by a distance of $n\epsilon$. %, and the decoder performs identification using a typicality test. 

We extensively use the method of types \cite[Ch.~2]{CK82}. Here a brief review of the definitions for type classes and $\delta$-typical sets is given. The type $\hat{P}_{x^n}$ of a given sequence $x^n$ is defined as the empirical distribution $\hat{P}_{x^n}(a)=N(a|x^n)/n$ for $a\in\X$, where $N(a|x^n)$ is the number of occurrences of the symbol $a\in\X$ in the sequence $x^n$. The space of all types over $\X$ of sequences of length $n$ is denoted by $\P_n(\X)$. The $\delta$-typical set $\T_\delta(p_X)$ is defined as the set of sequences $x^n\in\X^n$ such that for every $a\in\X$: $|\hat{P}_{x^n}(a)-p_X(a)|\leq\delta$  if $p_X(a)>0$, and $\hat{P}_{x^n}(a)=0$ if $p_X(a)=0$.
A type class is denoted by $\T(\hat{P})=\{ x^n \,:\; \hat{P}_{x^n}=\hat{P} \}$.
%%%
Similarly, a joint type is denoted by $\hat{P}_{x^n,y^n}(a,b)=N(a,b|x^n,y^n)/n$ for $(a,b)\in\X\times\Y$, where $N(a,b|x^n,y^n)$ is the number of occurrences of the symbol pair $(a,b)$ in the sequence $(x_i,y_i)_{i=1}^n$, and as a conditional type by 
$\hat{P}_{y^n|x^n}(b|a)=N(a,b|x^n,y^n)/N(a|x^n)$.
The conditional $\delta$-typical  set $\T_\delta(p_{Y|X}|x^n)$ is defined as the set of sequences $y^n\in\Y^n$ such that for every $b\in\Y$: $|\hat{P}_{y^n|x^n}(b|a)-p_{Y|X}(b|a)|\leq\delta$  if $p_{X,Y}(a,b)>0$, and $p_{X,Y}(a,b)=0$ if $p_X(a)=0$.

%%%
%%%%%%%%%%%%%%%%%%%%%%%%%%%%%%%%%%%%%%%%%%%%%%%%%%%%%%%%%%
\subsection*{The Codebook}
\label{Subsec.CodebookGeneration}

First, we show that there exists a code  such that the codewords are separated by a distance of $n\epsilon$. 
Let $p_X(x)$ be an input
distribution on $\X$, such that
%%%
\begin{equation}
    \label{Capacity}
    \mathbb{E} \left\{ \phi(X) \right\} = \sum_{x\in\X}
    p_X(x)\phi(x)
    \leq A-\epsilon'(\delta) \;,\,
\end{equation}
%%%
for $X\sim p_X(x)$, where $\epsilon'(\delta)\rightarrow 0$ as $\delta \to 0$. We may assume without loss of generality that $p_X$ is a type, due to the entropy continuity lemma \cite[Lem.~2.7]{CK82}. 
% ---
\begin{lemma}
    \label{Lem.DeterministicCodebook}
    Let $R<H(X)$. Then,
    %there exists $n_0(\epsilon,R)$, such that for every $n\geq n_0(\epsilon,R) $
    %
    for sufficiently small $\epsilon\in (0,1)$ and sufficiently large $n$, there exists a codebook $\U^*=\{v_i \,,\; i\in\M\}$, which consists of $|\M|$ sequences in $\X^n$, such that the following hold:
    \begin{enumerate}
        \item 
        All the codewords belong to the type class $\T(p_X)$, namely 
        %%%
        \begin{align}
            v_i \in \T(p_X) \text{ for all } i\in\M \;.\,
        \end{align}
        %%%
        \item
        The codewords are distanced by $n\epsilon$, i.e.,
        %%%
        \begin{align}
         \label{Eq.HammingdistanceProperty}
         d_H(v_i,v_j) \geq n\epsilon \text{ for all } i\neq j \;.\,
        \end{align}
        %%%
        \item
        The codebook size is at least $\frac{1}{2}\cdot 2^{nR}$, that is, 
        $|\M|\geq 2^{n(R-\frac{1}{n})}$.
    \end{enumerate}
    %for all $\epsilon\in (0,\epsilon_0]$.
\end{lemma}
%%%
\begin{proof}[Proof of Lemma~\ref{Lem.DeterministicCodebook}]
Denote
% ---
\begin{align}
    M \triangleq 2^{nR} \;.\,
\end{align}
% ---
Let $U_1,...,U_M$ be independent random sequences, each uniformly distributed over the type class of 
$p_X$, i.e.,
\begin{align}
    \Pr \left( U_i = x^n \right) =
    \begin{cases}
        \frac{1}{\left| \T(p_X) \right|} & x^n \in \T(p_X) \;,\, \\
        0 & x^n \notin \T(p_X) \;.\,
    \end{cases}
\end{align}
% i.i.d. and distributed according to \\
% $$\Pr(U_i = x^n) = p_X^n(x^n)=\prod_{t=1}^n p_X(x_t),$$
% where $x^n = (x_1,...,x_n) \in \X^n$. 
Next, define a new collection of sequences $V_1,...,V_M$ as follows,
%%%
\begin{equation}
    \label{Eq.HammingDistanceProperty}
     V_i =
        \begin{cases}
            U_i & \quad \text{if $d_H(U_i,U_j) \geq n\epsilon \quad \forall i \neq j$} \;,\, \\
            \emptyset & \quad \text{otherwise} \;,\,
        \end{cases}
\end{equation}
%%%
where $d_H(\cdot,\cdot)$ denotes the Hamming distance, and $\emptyset$ represents an idle sequence of no interest. The assignment $V_i = \emptyset$ is interpreted as ``dropping the $i$th word $U_i$."
%%%
Consider the following message set, 
% ---
\begin{equation}
\label{Eq.OperationalCodebook}
\widetilde{\mathcal{M}} = \left\{ i \, : \, V_i \neq \emptyset, i \in [\![M]\!] \right\} \;,\,
\end{equation}
% ---
corresponding to words that were not dropped, where we use the notation 
$\widetilde{\M}$ to indicate that the set is random.
%%%
%We refer to $\widetilde{\M}$ as the message set.
% Observe that due to the inherent randomness of the random sequences $V_i$, for $i \in [\![M]\!]$, the  set $\widetilde{\M}$ is random. 
% Now denote the resulting codebook by %is defined as the set of codewords 
% \[
% \mathscr{U}=\{V_i: i \in \widetilde{\M} \},
% \]
% hence the codebook is random as well. The probability distribution of $\mathscr{U}$ is then denoted by $P_\mathscr{U}$, namely
% \begin{align}
% P_{\mathscr{U}}(\U)=\Pr(\mathscr{U}=\U)
% \end{align}
% where we denote the random codebook by $\mathscr{U}$ and its realization by $\U$.
% To emphasize the dependence on the codebook, we will denote the type I and type II error probabilities by 
% $P_{e,1}(i,\mathscr{U})$ and $P_{e,2}(i,j,\mathscr{U})$, respectively. Notice that since the codebook is random, $P_{e,1}(i,\mathscr{U})$ and $P_{e,2}(i,j,\mathscr{U})$ are random variables.

%\subsection*{Rate Analysis}
We show that even though we removed words from the original collection $\{U_i\}_{i \in [\![M]\!]}$ (of size $M$), the rate decrease can be made negligible. Following the lines of \cite{A80}, we derive an upper-bound on $\Pr( |\widetilde{\M}| \leq \frac{1}{2}M)$ where $\widetilde{\M}$ defined in (\ref{Eq.OperationalCodebook}) is the operational message set. To this end, we will use the following concentration lemma,
%%%
\begin{lemma}[also in {\cite{A80}}]
\label{Lem.BernsteinIneq}
Let $A_1,\ldots,A_K$ be a sequence of discrete random variables. Then,
%%%
\begin{align}
    %\textbf{(a)} ~ 
    \Pr \left( \frac{1}{K} \sum_{i=1}^K A_i \geq c \right) \leq 2^{- cK} \prod_{i=1}^K \underset{a^{i-1}}{\max~} \mathbb{E} \left( 2^{A_i} \, \big| \, A^{i-1} = a^{i-1} \right) \;.\,
\end{align}
%%%
\end{lemma}
%%%

Now, define an indicator for dropping the $i$th word by
 %%%
 \begin{equation}
   \hat{V_i} =
   \begin{cases}
          1 & V_i = \emptyset \;,\, \\
          0 & V_i \neq \emptyset \,,\,
   \end{cases}
 \end{equation}
 %%%
 and notice the equivalence between the following events,
 %%%
 \begin{equation}
     \left\{\left| \widetilde{\M} \right| \leq \frac{1}{2} M\right\} = \left\{\sum_{i=1}^M \hat{V_i} > \frac{1}{2} M \right\} \;.\,
 \end{equation}
 %%%
Observe that
 $\hat{V_i} = 1$, if and only if $U_i$ is inside an $\epsilon$-sphere of some other $U_j$. Namely, $\hat{V_i} = 1$ iff $U_i \in \underset{j \neq i}{\bigcup} \S_{\epsilon}(U_j)$.
 The selection of codewords can be viewed as an iterative procedure. Specifically,
define
  \begin{equation}
   A_i =
   \begin{cases}
          1 & U_i \in \underset{j < i}{\bigcup} \S_{\epsilon}(U_j) \;,\, \\
          0 & \text{otherwise} \;,\,
   \end{cases}
   \label{Eq.Ai}
 \end{equation}
 %%%
 \begin{equation}
   B_i =
   \begin{cases}
          1 & U_i \in \underset{j > i}{\bigcup} \S_{\epsilon}(U_j) \;,\, \\
          0 & \text{otherwise} \;.\,
   \end{cases}
 \end{equation}
 % ---
Now, since $\hat{V}_i = 1$ implies that either $A_i=1$ or $B_i=1$, it follows that the number of dropped messages is bounded by
% ---
 \begin{align}
     M - \left| \widetilde{\M} \right| & = \sum_{i=1}^M \hat{V}_i 
     \nonumber\\
     & \leq \sum_{i=1}^M A_i + \sum_{i=1}^M B_i \;.\,
     \label{Eq.DtMAB}
 \end{align}
 % ---
Consider the event that
 \begin{align}
 \sum_{i=1}^M \hat{V}_i > \frac{1}{2} M \;.\,
 \label{Eq.DhVhalfM}
 \end{align}
If this holds, then the two sums in the right hand side of (\ref{Eq.DtMAB}) cannot be smaller than $\frac{1}{4} M$ together, that is, either 
$\sum_{i=1}^M A_i\geq\frac{1}{4} M$, or $\sum_{i=1}^M B_i\geq\frac{1}{4} M$, or both. Hence,
\begin{align}
   \left\{ \sum_{i=1}^M \hat{V}_i > \frac{1}{2} M \right\} \subseteq 
    \left\{\sum_{i=1}^M A_i \geq \frac{1}{4} M \right\} \cup
    \left\{\sum_{i=1}^M B_i \geq \frac{1}{4} M \right\} \;,\,
\end{align}
 %%%
 %%%
and by the union bound,
%%%
 \begin{align}
 \Pr\Big(\sum_{i=1}^M \hat{V_i} > \frac{1}{2} M \Big) &\leq \Pr\Big(\sum_{i=1}^M A_i \geq \frac{1}{4} M \Big) + \Pr\Big(\sum_{i=1}^M B_i \geq \frac{1}{4} M \Big)
 \nonumber\\
 &=2\Pr\Big(\sum_{i=1}^M A_i \geq \frac{1}{4} M \Big) \;,\,
 \end{align}
%%%
where the last line follows by symmetry, as the random variables $\bar{A}=\sum_{i=1}^M A_i$ and $\bar{B}=\sum_{i=1}^M B_i$ have the same probability distribution.
 
Next we apply Lemma~\ref{Lem.BernsteinIneq},
%%%
\begin{align}
 \Pr\Big(\sum_{i=1}^{M} A_i \geq \frac{1}{4}M\Big) &\leq 2^{-\frac{1}{4}M} \prod_{i=1}^M \underset{a^{i-1}}{\max}~ \mathbb{E} \left( 2^{A_i}|A^{i-1}=a^{i-1} \right) 
        % \nonumber \\
        %  &= 2^{-\frac{1}{4}M} \prod_{i=1}^M \underset{a^{i-1}}{\max}~\Big( 1\cdot\Pr\Big\{A_i= 0|A^{i-1}=a^{i-1}\Big\} + 2\cdot\Pr\Big\{A_i= 1|A^{i-1}=a^{i-1}\Big\}\Big) 
        %  \nonumber\\
        %  &\leq
        %  2^{-\frac{1}{4}M} \prod_{i=1}^M \Big( 1 + 2\cdot \underset{a^{i-1}}{\max}~\Pr(A_i = 1 | A^{i-1}=a^{i-1})\Big)
         \;.\,
     \label{Eq.DHammingUnionMeasure0}
     \end{align}
%%%
Consider the conditional expectation above.
Using the law of total expectation, we can add conditioning on $U^{i-1}$ as well, i.e.
\begin{align}
    \mathbb{E} \left( 2^{A_i}|A^{i-1}=a^{i-1} \right)
    % &=
    %  \mathbb{E}\left[\mathbb{E}(2^{A_i}|U^{i-1},A^{i-1}=a^{i-1})|A^{i-1}=a^{i-1}\right]
    %  \nonumber\\
     &=\sum_{u^{i-1}} \Pr( U^{i-1}=u^{i-1} |A^{i-1}=a^{i-1})\cdot \mathbb{E}(2^{A_i}|U^{i-1}=u^{i-1},A^{i-1}=a^{i-1})
     \nonumber\\
     &=\sum_{u^{i-1}} \Pr( U^{i-1}=u^{i-1} |A^{i-1}=a^{i-1})\cdot \mathbb{E}(2^{A_i}|U^{i-1}=u^{i-1})
     \nonumber\\
     &\leq\max_{u^{i-1}} \mathbb{E}(2^{A_i}|U^{i-1}=u^{i-1}) \;,\,
     \label{Eq.DHammingUnionMeasure1}
\end{align}
where the second equality holds since $A_i$, is a deterministic function of $U^{i-1}$ (see (\ref{Eq.Ai})).
Hence, by (\ref{Eq.DHammingUnionMeasure0})-(\ref{Eq.DHammingUnionMeasure1}),
%
% Furthermore, the expectation can be expanded as
% \begin{align}
%  \mathbb{E}(2^{A_i}|U^{i-1}=u^{i-1}) 
%          &=  1\cdot\Pr\Big\{A_i= 0|U^{i-1}=u^{i-1}\Big\} + 2\cdot\Pr\Big\{A_i= 1|U^{i-1}=u^{i-1}\Big\} 
%          \nonumber\\
%          &\leq
%           1 + 2\cdot \underset{a^{i-1}}{\max}~\Pr(A_i = 1 | A^{i-1}=a^{i-1}).
%      \label{Eq.DHammingUnionMeasure}
%      \end{align}
%
\begin{align}
 \Pr\Big(\sum_{i=1}^{M} A_i \geq \frac{1}{4}M\Big) &\leq 2^{-\frac{1}{4}M} \prod_{i=1}^M \underset{u^{i-1}}{\max}~ \mathbb{E} \left( 2^{A_i}|U^{i-1}=u^{i-1} \right) 
        \nonumber \\
         &= 2^{-\frac{1}{4}M} \prod_{i=1}^M \underset{u^{i-1}}{\max}~\Big( 1\cdot\Pr\Big\{A_i= 0|U^{i-1}=u^{i-1}\Big\} + 2\cdot\Pr\Big\{A_i= 1|U^{i-1}=u^{i-1}\Big\}\Big) 
         \nonumber\\
         &\leq
         2^{-\frac{1}{4}M} \prod_{i=1}^M \Big( 1 + 2\cdot \underset{u^{i-1}}{\max}~\Pr(A_i = 1 | U^{i-1}=u^{i-1})\Big) \;.\,
     \label{Eq.DHammingUnionMeasure}
     \end{align}

We bound the probability term $\Pr(A_i = 1 | U^{i-1}=u^{i-1})$, as follows. 
For a Hamming sphere of radius $n\epsilon$, 
%%%
\begin{align}
    \left| S_\epsilon(x^n) \right| \leq \binom{n}{n\epsilon}\cdot |\X|^{n\epsilon}  \leq 2^{n\theta(\epsilon)} \;,\,
\end{align}
%%%
for sufficiently large $n$, where 
% ---
\begin{align}
    \theta(\epsilon)=H_2(\epsilon) + \epsilon\log |\X| \;,\,
\end{align}
% ---
tends to zero as $\epsilon\to 0$. 
The first inequality holds by a simple combinatoric argument. Namely, counting the number of sequences with up to 
$n\epsilon$ different entries compared to a given $x^n$, we have $\binom{n}{n\epsilon}$ optional choices for the locations of those entries, and $|\X|$ possible values for each of those entries.
The last inequality follows from  Stirling's approximation \cite[Example 11.1.3]{Cover}.
Hence,
\begin{align}
    \left|\bigcup_{j=1}^M \S_{\epsilon}(u_j) \right| & \leq M 2^{n\theta(\epsilon)}
    \nonumber\\
    & = 2^{n \left( R + \theta(\epsilon) \right)} \;,\,
    \label{Eq.DSphereVolume}
\end{align}
for every given collection of sequences, $u_1,\ldots,u_M\in\T(p_X)$.
Consider a random sequence $\bar{X}^n$ that is uniformly distributed over the type class $\T(p_X)$, and statistically independent of $U_1,\ldots,U_M$. We use this external sequence as an auxiliary in the derivation below.
Then, %for every given collection of sequences, $u_1,\ldots,u_M\in\T(p_X)$,
\begin{align}
    \Pr\left( A_i = 1 \, \big| \, U^{i-1}=u^{i-1} \right)
    &=%\stackrel{(a)}{=}
    \Pr\left(U_i \in \bigcup_{j<i} \S_{\epsilon}(u_j) \right)
    \nonumber\\
    &=%\stackrel{(b)}{=}
   \Pr\left(\bar{X}^n \in \bigcup_{j<i} \S_{\epsilon}(u_j) \right)
   \nonumber\\
   &\leq %\stackrel{(c)}{\leq}  
   \Pr \left\{ \bar{X}^n\in \bigcup_{j=1}^M \S_{\epsilon}(u_j)\right\} \;.\,
    %\nonumber\\
    % &= 
    % \sum_{x^n\in  \T(p_X )~\cap~ \bigcup_{j=1}^M  \S_{\epsilon}(u_j) }  \frac{1}{|\T(p_X )|}
    %  \nonumber\\
    % &=
    % \frac{1}{|\T(p_X )|}\cdot \left| \T(p_X )~\cap~ \bigcup_{j=1}^M  \S_{\epsilon}(u_j) \right| 
    %      \nonumber\\
    % &\leq 
    % \frac{2^{n(R+ \theta(\epsilon))}}{|\T(p_X )|}
    % \nonumber\\
    % &\stackrel{(d)}{\leq} (n+1)^{|\X|}\cdot \frac{2^{n(R+ \theta(\epsilon))}}{2^{nH(X)}}
    % \nonumber\\
    % &\leq 
    % 2^{-n(H(X)-R- 2\theta(\epsilon))}
    \label{Eq.PrAiUim1}
\end{align}
%%%
%for sufficiently large $n$.
The first equality follows from the definition of $A_i$ in (\ref{Eq.Ai}) and because $U_1,\ldots,U_M$ are statistically independent. The second equality holds because $U_i$ and $\bar{X}^n$ are both uniformly distributed over the type class of $p_X$. The inequality follows as $\Pr(\F_1)\leq\Pr(\F_1\cup\F_2)$ for every pair $\F_1$, $F_2$ of probabilistic events.
Since  $\bar{X}^n$ is uniformly distributed over $\T(p_X)$, we have
% ---
\begin{align}
    \Pr \left\{\bar{X}^n\in \bigcup_{j=1}^M \S_{\epsilon}(u_j)\right\}
    & = 
    \sum_{x^n\in \T(p_X )~\cap~ \bigcup_{j=1}^M \S_{\epsilon}(u_j) }  \frac{1}{\left| \T(p_X ) \right|}
    \nonumber\\
    & =
    \frac{1}{\left| \T(p_X ) \right|} \cdot \left| \T(p_X )~\cap~ \bigcup_{j=1}^M \S_{\epsilon}(u_j) \right| 
    \nonumber\\
    &\leq 
    \frac{2^{n(R+ \theta(\epsilon))}}{\left| \T(p_X ) \right|}
    \nonumber\\
    &\leq %\stackrel{(d)}{\leq} 
    (n+1)^{\left| \X \right|}\cdot \frac{2^{n(R+ \theta(\epsilon))}}{2^{nH(X)}}
    \nonumber\\
    &\leq 
    2^{-n(H(X)-R- 2\theta(\epsilon))} \;,\,
    \label{Eq.PrAiUim1H}
\end{align}
% ---
for sufficiently large $n$,
where the first inequality follows from (\ref{Eq.DSphereVolume}), and the second is due to standard type class properties \cite[Th.~11.1.3]{Cover}.
 The last expression tends to zero as $n\rightarrow\infty$, provided that
%%%
\begin{align}
    R < H(X)-3\theta(\epsilon) \;.\,
\end{align}
%%%
Together with (\ref{Eq.PrAiUim1})-(\ref{Eq.PrAiUim1H}), this implies 
% ---
\begin{align}
    \label{Ineq.Result}
    \Pr \left( A_i = 1 \, \big| \, U^{i-1} = u^{i-1} \right) \leq 2^{-n\theta(\epsilon)} \;.\,
\end{align}
% ---
Now plugging (\ref{Ineq.Result}) into (\ref{Eq.DHammingUnionMeasure}) yields
% ---
     \begin{align}
         \Pr \left( \sum_{i=1}^{M} A_i \geq \frac{1}{4}M \right) & \leq  2^{-\frac{1}{4}M} \left( 1 + 2\cdot2^{-n\theta(\epsilon)} \right)^M
         \nonumber\\
         & = \left( 2^{-\frac{1}{4}} + 2^{\frac{3}{4}}\cdot2^{-n\theta(\epsilon)} \right)^M \;,\,
     \end{align}
% ---
for sufficiently large $n$, we have $2^{\frac{3}{4}}\cdot 2^{-n\theta(\epsilon)}\leq 2^{-5}$ hence,
 % ---
 \begin{align}
    2^{-\frac{1}{4}} + 2^{\frac{3}{4}} \cdot 2^{-n\theta(\epsilon)} & \leq 2^{-\frac{1}{4}}+2^{-5}
    \nonumber\\
    & = 0.8721 
    \nonumber\\
    & < 1 \;.\,
 \end{align}
 % ---
Thus we have a double exponential bound
%%%
 \begin{align}
    \label{Ineq.DRateError}
     \Pr\left( \left| \widetilde{\M} \right| \leq \frac{1}{2}M \right) & \leq
     2^{-\alpha_1 M} 
     \nonumber\\
     & = 2^{-\alpha_1 2^{nR}} \;,\,
 \end{align}
%%%
for some $\alpha_1>0$. We deduce that there exists at least one codebook with the desired properties. This completes the proof of Lemma~\ref{Lem.DeterministicCodebook}.
\end{proof}
% ---
We continue to the main part of the achievability proof.
Let $\U^*=\{ v_i \,,\; i\in\M \}$ be a codebook of size $2^{n(R-\frac{1}{n})}$ as in Lemma~\ref{Lem.DeterministicCodebook}.
Consider the following DI coding scheme for $\W$.

\subsubsection*{Encoding}
Given a message $i\in\M$ at the sender, transmit $x^n=v_i$.

\subsubsection*{Decoding}
Let $\delta > 0$, such that $\delta\rightarrow 0$ as $\epsilon\rightarrow 0$. Let $j\in\M$ be the message that the decoder wishes to identify. To do so, the decoder
%Then he simply 
checks whether the channel output $y^n$ belongs to the corresponding decoding set $\D_j$ or not, where
% ---
\begin{align}
    \D_j = \left\{ y^n \;:\, (v_j,y^n) \in \T_{\delta}(p_X W) \right\} \;.\,
    \label{Eq.DjTypicality}
\end{align}
% ---
Namely, given the channel output $y^n\in\Y^n$, if
$(v_j,y^n) \in \T_{\delta}(p_XW)$,
then the decoder declares that the message $j$ was sent.
On the other hand, if $(v_j,y^n) \notin \T_{\delta}(p_X W)$, it declares that $j$ was not sent.
% --------------------------
\subsection*{Error analysis}
\label{Subsec.AveErrorAnalysis}
% As an intermediate step, we first bound the following average error probabilities
% %%%
% \begin{align}
%     \label{Eq.Dbarp_1}
%   \bar{p}_{e,1}&\triangleq \frac{1}{|\M|} \sum_{i'\in\M}  P_{e,1}(i',\sU) \\
%     \label{Eq.Dbarp_2}
%   \bar{p}_{e,2}&\triangleq \frac{1}{|\M|} \sum_{i'\in\M} \max_{j'\in \M \,:\; j'\neq i'} P_{e,2}(i',j',\sU).
% \end{align}
% %%%
% The average error probabilities are associated with a uniformly distributed message. 
% The errors above involve two kinds of averages, first we have an arithmetic average over the operational message set $\widetilde{\M}$, and then we have a probabilistic average over the ensemble of codebooks. 
% Due to the symmetry of the codebook generation,
% the summands are identical, i.e.,
% $\mathbb{E}_{\sU}[P_{e,1}^{(n)}(i',\sU)] = \mathbb{E}_{\sU}[ P_{e,1}^{(n)}(1,\sU)]$ and 
% $\mathbb{E}_{\sU}[ \max_{ j'\neq i'} P_{e,2}^{(n)}(i',j',\sU)] = \mathbb{E}_{\sU}[ \max_{ j\neq 1} P_{e,2}^{(n)}(1,j,\sU)]$, for all 
% $i'\in \widetilde{\M}$.
% Therefore, in the analysis of the average error probabilities,  we may assume without loss of generality that a particular message $i$ was sent, as
% %%%
% \begin{align}
%     \label{Eq.DExpectedTypeIError}
%   \bar{p}_{e,1}&= \mathbb{E}_{\sU} [    P_{e,1}(i,\sU) ]\\
%   \bar{p}_{e,2}&= \mathbb{E}_{\sU}\left[ \max_{j\in \widetilde{M} \,:\; j\neq i} P_{e,2}(i,j,\sU) \right].
%   \label{Eq.DExpectedTypeIIError}
% \end{align}
% %%%

First, consider the error of type I, i.e., the event that $Y^n\notin \D_i$.
%%%
For every $i\in\M$,
the probability of identification error of type I, %%%
$%$\begin{align}
P_{e,1}(i) =\Pr((v_i,Y^n) \notin \T_{\delta}(p_X W))
     \label{Eq.DTypeIUnion}
$
tends to zero by standard type class considerations %More explicitly, by the conditional typicality lemma 
\cite[Th.~1.2]{K08}. %, we have: $P_{e,1}(i) = \Pr(Y^n\notin \T_\delta(W|v_i) |x^n=v_i) \leq e^{-\alpha_1(\delta)n}$ which tends to zero for sufficiently large $n$. The exponent $\alpha_1(\delta) = \frac{\delta^2}{8+4\delta}\epsilon(\delta)$ where $\epsilon(\delta)$ is a vanishing upper-bound as $\delta$ tends to zero.

We move to the error of type II, i.e., when $Y^n\in \D_j$ for $j\neq i$. 
To bound the probability of error $P_{e,2}(i,j)$, we use the conditional type-class intersection lemma, due to Ahlswede \cite{A80}, as stated below.
%%%%%%%%%%%%%%%%%%%%%%%%%%%%%%%%%%%%%%%%%%%%%%%%%%%%%%%%%%%%%
\begin{lemma}[see {\cite[Lem.~$I_1$]{A80}}]
\label{Lem.GenSeqInter}
Let $W:\X\rightarrow\Y$ be a channel  matrix of a DMC $\W$ with distinct rows. Then, for every $x^n,x'^n\in\T_\delta(p_X)$ with $d_H(x^n,x'^n)\geq n\epsilon$,
%%%
\begin{align}
    \frac{|\T_{\delta}(p_{Y|X}|x^n) \cap \T_{\delta}(p_{Y|X}|x'^n)|}{|\T_{\delta}(p_{Y|X}|x^n)|} \leq 2^{-nL(\epsilon)} \;,\,
\end{align}
with $p_{Y|X}\equiv W$,
%%%
for sufficiently large $n$ and some positive % and monotonically increasing
function $L(\epsilon) > 0$ which is independent of $n$.
\end{lemma}
%%%
% Now, for every $j\neq i$,
% %%%
% \begin{align}
% P_{e,2}(i,j) =&
% \Pr(\D_j|x^n=v_i)
% \label{Eq.DTypeIIError0}
% \end{align}
%%%
Now, for short notation, denote the conditional $\delta$-typical set in $\Y^n$, given $x^n\in\T(p_X)$, by 
%%%
\begin{align}
    \G(x^n) & \equiv \T_\delta(W|x^n)
    \nonumber\\
    & = \left\{ y^n \;:\, (x^n,y^n) \in  \T_{\delta}(p_X W) \right\} \;.\,
\end{align}
%%%
Then, for every $i\neq j$,
% Thus, by (\ref{Eq.DTypeIIError}), %and (\ref{Ineq.LLN}):
%%%
\begin{align}
    P_{e,2}(i,j) &=\Pr(\D_j|x^n=v_i)
    % = \sum_{y^n \,:\; (v_j,y^n)\in \T_\delta(p_X W)} W^n(y^n|v_i)
    \nonumber\\
    &=
    \sum_{y^n\in\G(v_j)} W^n(y^n|v_i)
    \nonumber\\
    &=\sum_{y^n \in \G(v_j) \cap \G(v_i)} W^n(y^n|v_i)
    + \sum_{y^n \in \G(v_j) \cap (\G(v_i))^c} W^n(y^n|v_i) \;.\,
    \label{Eq.D2TypeIIError}
\end{align}
%%%
Observe that the second sum in the last line is bounded by the probability
$\Pr(Y^n\notin \T_\delta(W|v_i) |x^n=v_i)$, which in turn is bounded by $ 2^{-\alpha_1(\delta)n}$ as before, and tends to zero as well.

To bound the first sum in 
(\ref{Eq.D2TypeIIError}), we first consider the cardinality of the set that the sum acts upon (the domain). We note that since $v_i$ and $v_j$ belong to the type class $\T(p_X)$ by the first property of Lemma~\ref{Lem.DeterministicCodebook}, it follows that they also belong to the $\delta$-typical set, i.e., $v_i,v_j\in\T_{\delta}(p_X)$. Further, according to the second property of Lemma~\ref{Lem.DeterministicCodebook}, every pair of codewords $v_i$ and $v_j$ satisfy $d_H(v_i,v_j)\geq n\epsilon $. Finally, having assumed that the rows of $W$ are distinct, we have by Lemma~\ref{Lem.GenSeqInter},
%%%
\begin{align}
    |\G(v_j) \cap \G(v_i)| &\leq 2^{-nL(\epsilon)}|\G(v_j)| \nonumber \\
    &\leq 2^{n[H(Y|X)-L(\epsilon)]} \;,\,
    \label{Eq.DIntersectionBound}
\end{align}
where $X\sim p_X$, as we explained below. The second inequality in (\ref{Eq.DIntersectionBound}) holds since the size of the conditional type class $\G(x^n)=\T_\delta(W|x^n)$ is bounded by $2^{nH(Y|X)}$ \cite[Lem.~2.5]{CK82}, as the type of $v_i$ and $v_j$ is $p_X$.
%%%
Furthermore, by standard type class properties \cite[Th.~1.2]{K08},
%%%
\begin{align}
    W^n(y^n|v_i) 
    &\leq 2^{-n[H(Y|X)-\delta\log|\Y|]} \;.\,
    \label{Eq.DConditionalProbBound}
\end{align}
%%%
Now by Equation~(\ref{Eq.DIntersectionBound}) and (\ref{Eq.DConditionalProbBound}),
%%%
\begin{align}
    \sum_{y^n \in \G(v_j) \cap \G(v_i)} W^n(y^n|v_i) \leq 2^{-n[L(\epsilon)-\delta\log|\Y|]} \;,\,
    \label{Eq.DConditionalProbBound2}
\end{align}
%%%
which tends to zero as $n\rightarrow\infty$ for sufficiently small $\delta>0$, such that $\delta\log|\Y|<L(\epsilon)$.
%%%
Thus, by  (\ref{Eq.D2TypeIIError}) and (\ref{Eq.DConditionalProbBound2}), the probability of type II error is bounded by
%%%
\begin{align}
\label{Ineq.DUppExpectedTypeIIError}
P_{e,2}(i,j) \leq 2^{-n\alpha_2(\epsilon,\delta)} \;,\,
\end{align}
%%%
for sufficiently large $n$, where $\alpha_2(\epsilon,\delta) = \min\{\alpha_1(\delta),L(\epsilon)-\delta\log|\Y|\}$.
%%%%%%%%%%%%%%%%%%%%%%%%%%%%%%%%%%%%%%%%
The proof follows by taking the limits $n\rightarrow\infty$, and $\epsilon$, $\delta\rightarrow 0$.
%%%%%%%%%%%%%%%%%%%%%%%%%%%%%%%%%%%%%%%%%%%%%%%%%%%
\subsection{Converse Proof}
\label{Subsec.Converse}
To prove the converse part, we will use the  following observation. Let $R>0$ be an achievable rate. 
We will assume to the contrary that there exist two different messages $i_1$ and $i_2$ that are represented by the same codeword, i.e., $u_{i_1}=u_{i_2}=x^n$, and show that this leads to error probabilities such that
% ---
\begin{align}
    P_{e,1}(i_1) + P_{e,2}(i_2,i_1)= 1 \;.\,
\end{align}
% ---
Hence the assumption is false. The number of messages $2^{nR}$ is thus bounded by the size of the subset of input sequences that satisfy the input constraint $\phi^n(x^n)\leq A$.
Then we notice that the average cost of a codeword depends only on its type, and hence this subset is in fact a union of type classes. This also implies that we have a strong converse for the DI capacity.

Consider a sequence of $(2^{nR},n,\lambda_1^{(n)},\lambda_2^{(n)})$ codes $(\U^{(n)},\D^{(n)})$ such that
$\lambda_1^{(n)}$ and $\lambda_2^{(n)}$ tend to zero as $n\rightarrow\infty$.
% --
\begin{lemma}
\label{Lem.DConverse}
Consider a sequence of codes  as described above. 
Then, given a sufficiently large $n$, the codebook 
$\U^{(n)}$ satisfies the following property.
There cannot be two distinct messages that are represented by the same codeword, i.e.,
% ---
\begin{align}
    i_1 \neq i_2 \; \quad \Rightarrow\quad u_{i_1} \neq u_{i_2} \;,\,
\end{align}
% ---
where $i_1,i_2\in[\![2^{nR}]\!]$.
\end{lemma}
% ---
\begin{proof}
    Assume to the contrary that 
    there exist two messages $i_1$ and $i_2$, where
    $i_1\neq i_2$, such that
    %%%
    \begin{align}
        u_{i_1}=u_{i_2}=x^n \;,\,
    \end{align}
    %%%
    for some $x^n\in\X^n$. Since $(\U^{(n)},\D^{(n)})$ form a $(2^{nR},n,\lambda_1^{(n)},\lambda_2^{(n)})$ code, we have
    % ---
    \begin{align}
        P_{e,1}(i_1) = W^n(\D_{i_1}^c|x^n) & \leq \lambda_1^{(n)} 
        \nonumber\\
        P_{e,2}(i_2,i_1) = W^n(\D_{i_1}|x^n) & \leq \lambda_2^{(n)} \;.\,
    \end{align}
    % ---
This leads to a contradiction as
%%%
    \begin{align}
        1 & = W^n(\D_{i_1}^c|x^n) + W^n(\D_{i_1}|x^n)
        \nonumber\\
        & = P_{e,1}(i_1) + P_{e,2}(i_2,i_1)
        \nonumber\\
        & \leq \lambda_1^{(n)} + \lambda_2^{(n)} \;.\,
    \end{align}
%%%
Hence, the assumption is false, and $i_1$ and $i_2$ cannot have the same codeword.
\end{proof}
%%%
By Lemma~\ref{Lem.DConverse}, each message has a distinct codeword. Hence, the number of messages is bounded by the number of input sequences that satisfy the input constraint. That is, the size of the codebook is upper-bounded as follows:
%%%
\begin{align}
    \label{Ineq.DConverse}
    2^{nR} \leq \left| \left\{x^n \;:\, \frac{1}{n} \sum_{t=1}^n
     \phi(x_t) \leq A\right\} \right| \;.\,
\end{align}
%%%
Notice that the input cost of a given sequence $x^n$ depends only on the type of the sequence, since
%%%
\begin{align}
    \frac{1}{n} \sum_{t=1}^n
    \phi(x_t) & = \sum_{a\in\X} \hat{P}_{x^n}(a)\phi(a)
    \nonumber\\
    & = \mathbb{E}\big\{\phi(X')\big\} \;,\,
\end{align}
%%%
where the random variable $X'$ is  distributed according to the type of $x^n$, i.e., $p_{X'}= \hat{P}_{x^n}$. Therefore, the subset on the right hand side of (\ref{Ineq.DConverse}) can be written as a union of type classes:
%%%
\begin{align}
     \left| \left\{x^n: \frac{1}{n} \sum_{t=1}^n
     \phi(x_t) \leq A \right\} \right| &= \;
     \Bigg|\bigcup_{\substack{p_{X'}\in\P_n(\X):\\\mathbb{E}\{\phi(X')\} \leq A}
     } \T(p_{X'})\Bigg|
     \nonumber\\
     &\leq \; \left| \P_n(\X) \right|
     \max_{\substack{p_{X'}\in\P_n(\X):\\\mathbb{E}\{\phi(X')\} \leq A}} \left| \T(p_{X'}) \right|
     \nonumber\\
     &\leq \left| \P_n(\X) \right| \cdot 2^{n H(X')}
     \nonumber\\
     &\leq 2^{n \left( H(X') + \alpha_n \right)}
     \nonumber\\
     & \leq 
     2^{n \left( \mathsf{C}_{DI}(\W) + \alpha_n \right)} \;,\,
     \label{Ineq.DEmpricalDis}
\end{align}
%%%
where $\alpha_n\rightarrow 0$ as $n\rightarrow\infty$,
where $\P_n(\X)$ denotes the space of all types over $\X$ of sequences of length $n$. 
The second inequality holds since  the size of a type class $\T(p_{X'})$ is bounded by $|\T(p_{X'})|\leq
2^{nH(X')}$ \cite[Th.~11.1.3]{Cover}. The third inequality holds since  the number of types on $\X$ is polynomial in $n$ \cite[Th.~11.1.1]{Cover}. Thus, by (\ref{Ineq.DConverse}) and (\ref{Ineq.DEmpricalDis}), the code rate is bounded by $R\leq \mathsf{C}_{DI}(\W)+\alpha_n$, which completes the proof of Theorem~\ref{Th.DDICapacity}.
\qed
% ----------------------------
\section{The Gaussian Channel}
\label{Sec.GaussianChannel}
In this section, we consider the Gaussian channel $\mathscr{G}$, specified by the input-output relation
%%%
\begin{align}
    \mathbf{Y}=\mathbf{x}+\mathbf{Z} \;.\,
\end{align}
%%%
with additive white Gaussian noise, i.e.,
when the noise sequence $\mathbf{Z}$ is i.i.d. $\sim \mathcal{N}(0,\sigma^2)$.  The transmission power is limited to $\norm{\textbf{x}}^2\leq nA$.
%%%
\begin{figure}[htb]
    \centering
	\tikzstyle{farbverlauf} = [ top color=white, bottom color=white!80!gray]
\tikzstyle{block1} = [draw,top color=white, bottom color=blue!20!white, rectangle, rounded corners,
minimum height=2em, minimum width=2.5em]

\tikzstyle{block2} = [draw,top color=white, bottom color=white!80!blue, rectangle, rounded corners,
minimum height=2em, minimum width=2.5em]

\tikzstyle{input} = [coordinate]
\tikzstyle{sum} = [draw, circle,inner sep=0pt, minimum size=5mm,  thick]
\tikzstyle{arrow}=[draw,->] %{Latex[length=3mm]},
\begin{tikzpicture}[auto, node distance=2cm,>=latex']
\node[] (M) {$i$};
\node[block1,right=.5cm of M] (enc) {Encoder};
\node[sum, right=1cm of enc] (channel) {$+$};
\node[block2, right=1cm of channel] (dec) {Decoder};
\node[below=.5cm of dec] (Target) {$j$};
\node[right=.5cm of dec] (Output) {$\text{\small Yes/No}$};
%m =_{\lambda_1,\lambda_2} m^*

% \node[above right=1cm of dec] (correctID) {$\textcolor{red}{\text{missed ID ($\lambda_1$)}}$};
% \node[right=1cm of dec] (missedID) {$\textcolor{black!60!green}{\text{correct ID}}$};
% \node[below right=1cm of dec] (falseID) {$\textcolor{red}{\text{false ID ($\lambda_2$)}}$};

\node[above=.7cm of channel] (noise) {$\textbf{Z}$};
\draw[->] (M) -- (enc);
\draw[->] (enc) --node[above]{$\textbf{u}_i$} (channel);
\draw[->] (noise) -- (channel);
\draw[->] (channel) --node[above]{$\textbf{Y}$} (dec);

% \draw[->] (dec) -- (correctID);
% \draw[->] (dec) -- (missedID);
% \draw[->] (dec) -- (falseID);

\draw[->] (dec) -- (Output);
\draw[->] (Target) -- (dec);
% \path[->] (dec) edge (Mhat);

\end{tikzpicture}
	\caption{Deterministic identification for the standard Gaussian channel}
	\label{Fig.GaussianChannel}
\end{figure}
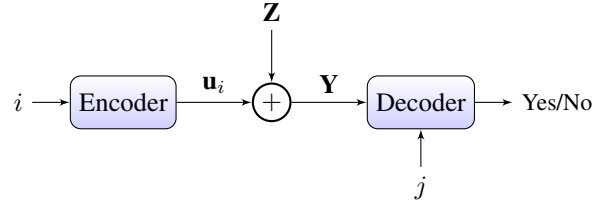
%%%
% ------------------------------------------
\subsection{Coding for the Gaussian Channel}
The definition of a DI code for the Gaussian channel is given below.
%%%
\begin{definition}[Gaussian DI Code]
\label{GdeterministicIDCode}
A $(2^{nR},n)$ DI code for a Gaussian channel $\sG$ under input constraint $A$, assuming $2^{nR}$ is an integer, is defined as a system $(\U,\mathscr{D})$ consisting of a codebook $\U=\{ \mathbf{u}_i \}_{i\in[\![2^{nR}]\!]}$, $\U\subset \X^n$, such that
%%%
\begin{align} 
    \norm{\mathbf{u}_i}^2 \leq nA \;,\,
\end{align}
%%%
for all $i\in[\![2^{nR}]\!]$ and a collection of decoding regions
$\mathscr{D}=\{ \D_i \}_{i\in[\![2^{nR}]\!]}$
with
% ---
\begin{align}
    \bigcup_{i=1}^{2^{nR}} \D_i \subset \mathbb{R}^n \;.\,
\end{align}
% ---
Given a message $i\in [\![2^{nR}]\!]$, the encoder transmits $\mathbf{u}_i$. The decoder's aim is to answer the following question: Was a desired message $j$ sent or not? There are two types of errors that may occur: 
Rejecting of the true message, or accepting a false message. Those are referred to as type I and type II errors, respectively.

The error probabilities of the identification code $(\U,\mathscr{D})$ are given by
%%%
\begin{align}
 P_{e,1}(i)&= 1-\int_{\D_i} f_{\mathbf{Z}}(\mathbf{y}-\mathbf{u}_i)\, d\mathbf{y}
 &&\hspace{-2cm}  \text{correctness property}
 \;,\,
 \label{Eq.GTypeIErrorDef}
 \\
 P_{e,2}(i,j)&= \int_{\D_j} f_{\mathbf{Z}}(\mathbf{y}-\mathbf{u}_i) \, d\mathbf{y} &&\hspace{-2cm} \text{disjointedness property}
 \;.\,
 \label{Eq.GTypeIIErrorDef}
\end{align}
%%%
with the noise formula given by
% ---
\begin{align}
   f_{\mathbf{Z}}(\mathbf{z})=\frac{1}{(2\pi\sigma^2)^{n/2}}
e^{-\norm{\mathbf{z}}^2/2\sigma^2} \;,\,
\end{align}
% ---
(see Figure~\ref{Fig.GeometricID}).
%%%
A $(2^{nR},n,\lambda_1,\lambda_2)$ DI code further satisfies
% ---
\begin{align}
    \label{Eq.GTypeIError}
    P_{e,1}(i) &\leq \lambda_1 \;,\,
    \\
    \label{Eq.GTypeIIError}
    P_{e,2}(i,j) &\leq \lambda_2 \;,\,
\end{align}
% ---
for all $i,j\in [\![2^{nR}]\!]$, such that
$i\neq j$.
%%%

A rate $R>0$ is called achievable if for every
$\lambda_1,\lambda_2>0$ and sufficiently large $n$, there exists a $(2^{nR},n,\lambda_1,\lambda_2)$ DI code. 
The operational DI capacity of the Gaussian channel is defined as the supremum of achievable rates, and will be denoted by $\mathbb{C}_{DI}(\mathscr{G})$. 
\end{definition}
%%%
%%%%%%%%%%%%%%%%%%%%%%%%%%%%%%%%%%%%%%%%%%%%%%%%%%%%%%%%%%%
\subsection{Main Result - Gaussian Channel}
%%%
Our DI capacity theorem for the Gaussian channel is stated below.
%%%
\begin{theorem}
\label{Th.GDICapacity}
The DI capacity of the Gaussian channel $\sG$ is given by
%%%
\begin{align}
    \label{Eq.GDICapacity}
    \mathbb{C}_{DI}(\sG) = \infty \;.\,
\end{align}
%%%
\end{theorem}
The proof of Theorem~\ref{Th.GDICapacity} is given below.
%%%
\begin{proof}
Consider the Gaussian channel $\mathscr{G}$.
To show that the capacity is infinite, it suffices to prove the direct part. We show here that the DI capacity of the Gaussian channel can be achieved using a simple distance-decoder.
A DI code for the Gaussian channel $\sG$ is constructed as follows. Since the decoder can normalize the output symbols by $\frac{1}{\sqrt{n}}$, we have an equivalent input-output relation,
%%%
\begin{align}
    \bar{\fY}=\bar{\fx}+\bar{\fZ} \;,\,
\end{align}
%%%
where the noise sequence $\bar{\fZ}$ is i.i.d. $\sim \mathcal{N}\left(0,\frac{\sigma^2}{n}\right)$, and an input power constraint 
%%%
\begin{align}
    \norm{\bar{\fx}} \leq \sqrt{A} \;,\,
\end{align}
with $\bar{\fx}=\frac{1}{\sqrt{n}}\fx$, $\bar{\fZ}=\frac{1}{\sqrt{n}}\fZ$, and $\bar{\fY}=\frac{1}{\sqrt{n}}\fY$.
%%%
\subsubsection*{Codebook construction}
\label{Subsec.CodebookConstructionGaussian}
Let $\mathscr{S}$ denote a sphere packing, i.e., an arrangement of $L$ non-overlapping spheres $\S_{\fu_i}(n,r_0)$, $i\in [\![L]\!]$, that cover a bigger sphere $\S_{\mathbf{0}}(n,r_1)$, with $r_1>r_0$. 
As opposed to standard sphere packing coding techniques, the small spheres are not necessarily entirely contained within the bigger sphere (see Figure~\ref{Fig.Density}). That is, we only require that the spheres are disjoint from each other and have a non-empty intersection with  $\S_{\mathbf{0}}(n,r_1)$.
%$n$-dimensional Euclidean space %
The packing density $\Delta_n(\mathscr{S})$ is defined as the fraction of the big sphere volume $\text{Vol}\left(\S_{\mathbf{0}}(n,r_1)\right)$ that is  covered by the small spheres, i.e.
\begin{align}
    \Delta_n(\mathscr{S}) \triangleq \frac{\text{Vol}\left(\S_{\f0}(n,r_1)\cap\bigcup_{i=1}^{L}\S_{\fu_i}(n,r_0)\right)}{\text{Vol}(\S_{\f0}(n,r_1))} \;,\,
    \label{Eq.DensitySphereFast}
\end{align}
(see \cite[Ch.~1]{CHSN13}).
A sphere packing is called \emph{saturated} if  no spheres can be added to the arrangement without overlap.
\begin{figure}[htb]
    \centering
	\scalebox{1}{

\begin{tikzpicture}[scale=.55][thick]

\draw[thick] (0,0) circle (3.1cm);

%% Entire Spheres
\draw (0,0) circle (1cm);
\draw [] (0,0) circle (1cm);

\draw (2,0) circle (1cm);
\draw [] (2,0) circle (1cm);

\draw (1,1.73) circle (1cm);
\draw [] (1,1.73) circle (1cm);

\draw (-1,1.73) circle (1cm);
\draw [fill=white, fill opacity=0.5] (-1,1.73) circle (1cm);

\draw (-2,0) circle (1cm);
\draw [fill=white, fill opacity=0.5] (-2,0) circle (1cm);

\draw (-1,-1.73) circle (1cm);
\draw [fill=white, fill opacity=0.3] (-1,-1.73) circle (1cm);

\draw (1,-1.73) circle (1cm);
\draw [fill=white, fill opacity=0.4] (1,-1.73) circle (1cm);

%% Partial Spheres

\draw (3,-1.73) circle (1cm);
\draw [fill=gray!30!white, fill opacity=0.4] (3,-1.73) circle (1cm);

% \node at (3,1.73) {$.$};
\draw (3,1.73) circle (1cm);
\draw [fill=gray!30!white, fill opacity=0.4] (3,1.73) circle (1cm);

\draw (0,2*1.73) circle (1cm);
\draw [fill=gray!30!white, fill opacity=0.4] (0,2*1.73) circle (1cm);

\draw (-3,1.73) circle (1cm);
\draw [fill=gray!30!white, fill opacity=0.4] (-3,1.73) circle (1cm);

\draw (-3,-1.73) circle (1cm);
\draw [fill=gray!30!white, fill opacity=0.4] (-3,-1.73) circle (1cm);

\draw (0,-2*1.73) circle (1cm);
\draw [fill=gray!30!white, fill opacity=0.4] (0,-2*1.73) circle (1cm);

%% Arrow
\draw (0,0) -- (3.1,0) node [right,font=\small] {$\sqrt{A}-\sqrt{\epsilon}$};

\draw (-3,-1.73)-- (-3.707,-2.437)    node [below,font=\small] {$\sqrt{\epsilon}\qquad$};

\end{tikzpicture}}
	\caption{Illustration of a sphere packing, where small spheres of radius $r_0=\sqrt{\epsilon}$ cover a bigger sphere of radius $r_1=\sqrt{A}-\sqrt{\epsilon}$. 
		The small spheres are disjoint from each other and have a non-empty intersection with  the big sphere.
		Some of the small spheres, marked in gray, are not  entirely contained within the bigger sphere, and yet they are considered to be a part of the packing arrangement.
	     As we assign a codeword to each small sphere center, the norm of a codeword is bounded by $\sqrt{A}$ as required.}
	\label{Fig.Density}
\end{figure}
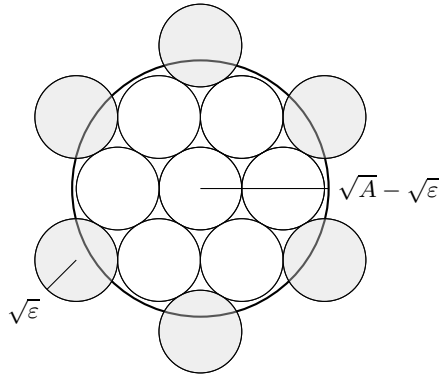
% ---
sphere packing is called \emph{saturated} if  no spheres can be added to the arrangement without overlap.

We use a packing argument that has a similar flavor as in the Minkowski--Hlawka theorem in lattice theory \cite{CHSN13}. We use the property that there exists an arrangement $\bigcup_{i=1}^{L} \S_{\fu_i}(n,\sqrt{\epsilon_n})$ of non-overlapping spheres inside $\S_{\f0}(n,\sqrt{A})$ with a density of $\Delta_n(\mathscr{S})\geq 2^{-n}$ \cite[Lem.~2.1]{C10}. Specifically, consider a saturated packing arrangement of $L(n,R)=2^{nR}$ spheres of radius $r_0=\sqrt{\epsilon}$ covering the big sphere $\S_{\f0}(n,r_1=\sqrt{A}-\sqrt{\epsilon})$, i.e., such that no spheres can be added without overlap. Then, for such an arrangement, there cannot be  a point in the big sphere $\S_{\f0}(n,r_1)$ with a distance of more than $2r_0$ from all sphere centers. Otherwise, a new sphere could be added. As a consequence, if we double the radius of each sphere, the  $2r_0$-radius spheres  cover the whole sphere of radius $r_1$. In general, the volume of a hyper-sphere of radius $r$ is given by
%%%
\begin{align}
    \text{Vol}\left(\S_\epsilon(\fx,r)\right)=\frac{\pi^{\frac{n}{2}}}{\Gamma(\frac{n}{2}+1)}\cdot r^{n} \;,\,
    \label{Eq.VolS}
\end{align}
%%%
(see Eq. (16) in \cite{CHSN13}).
Hence, doubling the radius multiplies the volume by $2^n$. Since the $2r_0$-radius spheres cover the entire sphere of radius $r_1$, it follows that the original $r_0$-radius packing has density at least $2^{ -n}$, i.e.,
%%%
\begin{align}
    \Delta_n(\mathscr{S})\geq 2^{-n} \;.\,
    \label{Eq.MinkowskiDeltaFast}
\end{align}
%%%

We assign a codeword to the center $\fu_i$ of each small sphere.
The codewords satisfy the input constraint as
$\norm{\fu_i}\leq r_0+r_1=\sqrt{A}$.
Since the small spheres have the same volume, the total number of spheres is bounded from below by
% ---
\begin{align}
    L & = \frac{\text{Vol}\left(\bigcup_{i=1}^{L}\S_{\fu_i}(n,r_0)\right)}{\text{Vol}(\S_{\fu_1}(n,r_0))}
    \nonumber\\
    & \geq\frac{\text{Vol}\left(\S_{\f0}(n,r_1)\cap\bigcup_{i=1}^{L}\S_{\fu_i}(n,r_0)\right)}{\text{Vol}(\S_{\fu_1}(n,r_0))}
    \nonumber\\
    & = \frac{\Delta_n(\mathscr{S})\cdot
    \text{Vol}(\S_{\mathbf{0}}(n,r_1)))}{\text{Vol}(\S_{\fu_1}(n,r_0))}
    \nonumber\\
    & \geq 2^{-n}\cdot \frac{
    \text{Vol}(\S_{\mathbf{0}}(n,r_1)))}{\text{Vol}(\S_{\fu_1}(n,r_0))}
    \nonumber\\
    & = 2^{-n}\cdot \frac{r_1^n}{r_0^n} \;,\,
\end{align}
% ---
where the second equality is due to (\ref{Eq.DensitySphereFast}),  the inequality that follows holds by (\ref{Eq.MinkowskiDeltaFast}), and the last equality follows from (\ref{Eq.VolS}).
That is, the codebook size satisfies
%%%
\begin{align}
     L(n,R) & = 2^{nR}
     \nonumber\\
     & \geq 2^{-n}\cdot\left(\frac{\sqrt{A}-\sqrt{\epsilon}}{\sqrt{\epsilon}}\right)^n \;.\,
\end{align}
%%%
%
Hence,
%%%
\begin{align}
    \label{Eq.RateFast}
    R\geq \frac{1}{2}\log\left(\frac{A}{\epsilon}\right)-1 \;.\,
\end{align}
\subsubsection*{Encoding}
Given a message $i\in [\![2^{nR}]\!]$, transmit $\bar{\fx}=\bar{\fu}_i$.
% -----------------------
\subsubsection*{Decoding}
Let $\delta>0$. To identify whether a message $j\in \M$ was sent, the decoder checks whether the channel output $\mathbf{y}$ belongs to the following decoding set, % $\D_j$ or not, where
%%%
\begin{align}
  %  \label{Eq.GDecodingSet}
    \D_j = \left\{ \bar{\fy}\in\mathbb{R}^n \,:\; \norm{\bar{\fy}-\bar{\fu}_j}
    \leq \sqrt{\sigma_Z^2+\delta} \right\} \;.\,
\end{align}
%%%
% -----------------------------
\subsubsection*{Error Analysis}
Consider the type I error, i.e., when the transmitter sends $\bar{\fu}_i$, yet $\bar{\fY}\notin\D_i$. For every $i\in[\![2^{nR}]\!]$, the type I error probability is bounded by
%%%
\begin{align}
    P_{e,1}(i)&= \Pr\left(\norm{\bar{\fY}-\bar{\fu}_i}^2 >  \sigma_Z^2+\delta \,\big|\, \bar{\fx}=\bar{\fu}_i \right)
    \nonumber\\
    &=\Pr\left(\norm{\bar{\fZ}}^2 > \sigma_Z^2 + \delta \right)
    \nonumber\\
    &=\Pr\left(\sum_{t=1}^n {\bar{Z}_t}^2> \sigma_Z^2 + \delta \right)
    \nonumber\\
    & \leq \frac{3\sigma_Z^4}{n\delta^2}
    \nonumber\\
    & \leq \lambda_1
    \;,\,
\end{align}
%%%
which tends to zero as $n\to\infty$, where the last inequality holds by Chebyshev's inequality.

Next, we address the type II error, i.e., when $\bar{\fY}\in\D_j$ while the transmitter sent $\bar{\fu}_i$.
Then, for every $i,j\in[\![2^{nR}]\!]$, where $i\neq j$, the type II error probability is given by
\begin{align}
    P_{e,2}(i,j)&= \Pr\left( \norm{\bar{\fY}-\bar{\fu}_j}^2\leq \sigma_Z^2+\delta \,\big|\, \bar{\fx}=\bar{\fu}_i \right)
    \nonumber\\
    &=\Pr\left( \norm{\bar{\fu}_i-\bar{\fu}_j+\bar{\fZ}}^2\leq \sigma_Z^2+\delta \right) \;.\,
    \label{Eq.Pe2G}
\end{align}
%%%
Observe that the square norm can be expressed as
%%%
\begin{align}
    \norm{\bar{\fu}_i-\bar{\fu}_j+\bar{\fZ}}^2=
     \norm{\bar{\fu}_i-\bar{\fu}_j}^2+\norm{\bar{\fZ}}^2+2\sum_{t=1}^n (\bar{u}_{i,t}-\bar{u}_{j,t})Z_t \;.\,
     \label{Eq.Pe2norm}
\end{align}
%%%
Then, define the event
%%%
\begin{align}
    \E_0 = \left\{ \left| \sum_{t=1}^n (\bar{u}_{i,t} - \bar{u}_{j,t}) \bar{Z}_t \right| > \frac{\delta}{2} \right\} \;,\,
\end{align}
%%%
By Chebyshev's inequality, the probability of this event vanishes, 
%%%
\begin{align}
    \Pr(\E_0) & \leq \frac{\sigma_Z^2\sum_{t=1}^n (\bar{u}_{i,t}-\bar{u}_{j,t})^2}{n\left( \frac{\delta}{2} \right)^2}
    \nonumber\\
    & = \frac{4 \sigma_Z^2\norm{\bar{\fu}_i - \bar{\fu}_j}^2}{n\delta^2}
    \nonumber\\
    & \leq
    \frac{16 \sigma_Z^2 A}{n\delta^2}
    \nonumber\\
    & \leq \zeta \;,\,
    \label{Eq.PeE0G}
\end{align}
%%%
for sufficiently large $n$, where $\zeta > 0$ is arbitrary constant, where the first inequality holds since the sequence $\{\bar{Z}_t\}$ is i.i.d. $\sim\mathcal{N}\left(0,\frac{\sigma_Z^2}{n}\right)$, and the second inequality follows as
% ---
\begin{align}
    \norm{\bar{\fu}_i - \bar{\fu}_j}^2 & \leq
    (\norm{\bar{\fu}_i}+\norm{\bar{\fu}_j})^2
    \nonumber\\
    & \leq (\sqrt{A}+\sqrt{A})^2
    \nonumber\\
    & = 4A \;,\,    
\end{align}
% ---
by the triangle inequality. Now let us define following event
% ---
\begin{align}
    \A_{i,j} \left( \sigma_Z^2 + \delta \right) \equiv \left\{ \bar{\fZ} \in \mathbb{R}^n \;:\, \norm{\bar{\fu}_i - \bar{\fu}_j + \bar{\fZ}}^2 \leq \sigma_Z^2 + \delta \right\} \;,\,
    \label{Ineq.Event_A_i_j}
\end{align}
% ---
Observe that given the complementary event $\E_0^c$, we have
% ---
\begin{align}
    2\sum_{t=1}^n (\bar{u}_{i,t}-\bar{u}_{j,t})\bar{Z}_t\geq - \delta \;,\,
\end{align}
% ---
hence, by (\ref{Eq.Pe2norm}), the event $\A_{i,j} \left( \sigma_Z^2 + \delta \right)$ implies following event
% ---
\begin{align}
    \E_1 = \left\{ \bar{\fZ} \in \mathbb{R}^n \;:\, \norm{\bar{\fu}_i-\bar{\fu}_j}^2+\norm{\bar{\fZ}}^2\leq \sigma_Z^2 + 2\delta \right\} \;.\,
    \label{Ineq.Event_E_1}
\end{align}
% ---
Applying the law of total probability to (\ref{Eq.Pe2G}), we have
% --
\begin{align}
    P_{e,2}(i,j) & \stackrel{(a)}{=}
    \Pr \left( \left\{ \A_{i,j} \left( \sigma_Z^2 + \delta \right) \right\} \cap \E_0 \right) + \Pr\left( \left\{ \A_{i,j} \left( \sigma_Z^2 + \delta \right) \right\} \cap \E_0^c \right)
    \nonumber\\
    & \stackrel{(b)}{\leq}
    \Pr(\E_0) + \Pr\left( \left\{ \A_{i,j} \left( \sigma_Z^2 + \delta \right) \right\} \cap \E_0^c \right)
    \nonumber\\
    & \stackrel{(c)}{\leq}
    \zeta + \Pr \left( \E_1 \right) \;,\,
\end{align}
%%%
where $(a)$ is due to (\ref{Ineq.Event_A_i_j}), $(b)$ holds since each probability is bounded by $1$ and $(c)$ follows from (\ref{Ineq.Event_E_1}).
Based on the codebook construction, each codeword is surrounded by a sphere of radius $\sqrt{\epsilon}$, which implies
% ---
\begin{align}
    \norm{\bar{\fu}_i-\bar{\fu}_j} \geq \sqrt{\epsilon} \;,\,
\end{align}
% ---
Hence,
% ---
\begin{align}
    - \norm{\bar{\fu}_i-\bar{\fu}_j}^2 \leq - \epsilon \;.\,
\end{align}
% ---
Thus, choosing $\delta = \frac{\epsilon}{3}$, we obtain
% ---
\begin{align}
    P_{e,2}(i,j)   
    &\leq
    \Pr\left( \norm{\bar{\fZ}}^2\leq \sigma_Z^2-\delta%+\delta+\beta-\epsilon 
    \right)+\zeta
    \nonumber\\
    &=\Pr\left( \sum_{t=1}^n \bar{Z}_t^2-\sigma_Z^2\leq -\delta%+\delta+\beta-\epsilon 
    \right)+\zeta
    \nonumber\\
    &\leq 
    \frac{\sum_{t=1}^n \text{Var}(\bar{Z}_t^2)}{\delta^2}+\zeta
    \nonumber\\
    &\leq \frac{n\cdot\mathbb{E}\{ \bar{Z}_t^4\}}{\delta^2}+\zeta
    \nonumber\\
    & = \frac{3\sigma_Z^4}{n\delta^2} + \zeta
    \nonumber\\
    & \leq \lambda_2 \;,\,
\end{align}
% ---
for sufficiently large $n$, where $\lambda_2 > 0$ is arbitrary constant, since the fourth moment of a Gaussian variable 
$V\sim \N(0,\sigma_V^2)$ is $\mathbb{E}\{V^4\}=3\sigma_V^4$.

We have thus shown that for every 
$\lambda_1,\lambda_2 > 0$ and sufficietnly large $n$,
there exists a 
 $(2^{nR}, n, \lambda_1, \lambda_2)$ code.
The proof follows by taking the limits $n\rightarrow\infty$, then $\gamma$, $\delta$ $\rightarrow 0$, hence
$\epsilon$, $\beta$ $\rightarrow 0$ and
$R\to\infty$ by (\ref{Eq.RateFast}).
\end{proof}

\subsection{Alternative Proof: Discretization}
In this subsection, we give a second proof for the DI capacity theorem of the Gaussian channel, Theorem~\ref{Th.GDICapacity}. We show that the theorem can be obtained from our result on the DMC in Theorem~\ref{Th.DDICapacity}, using discretization.
We show that given a Gaussian random variable $X \sim \N(0, A)$, the entropy of the discretized variable is approximately
% ---
\begin{align}
    \frac{1}{2}\log(2\pi eA)-\frac{2\Delta}{\sqrt{2\pi A}}+\log\frac{1}{\Delta} \;,\,
\end{align}
% ---
where $\Delta>0$ is the discretization step. Therefore, as $\Delta$ tends to zero, the discretized entropy grows to infinity. 

Our discretization procedure is similar to the one presented in \cite[see Sec. 3.4.1]{GK12}.
Consider a Gaussian random variable $X \sim \N(0, A)$, hence $h(X) = \frac{1}{2}\log(2\pi e A)$. Let $J>0$ be arbitrarily large and $\Delta>0$ be arbitrarily small. Consider the discretized variable
%%%
\begin{align}
    \hX \in \left\{-J \Delta,\,-(J-1) \Delta,\,\cdots,\,-\Delta,~0,\,\Delta,\,\cdots,\,(J-1) \Delta,\,J\Delta \right\} \;,\,
\end{align}
%%%
obtained by mapping $X$ to the closest discretization point $\hX = g_{J,\Delta}(X)$, such that $|\hX| \leq |X|$. Clearly, $\mathbb{E}(\hX^2) \leq \mathbb{E}(X^2) = A$. More specifically, 
% ---
\begin{align}
    \label{Eq.Xhat}
    g_{J,\Delta}(x) =
    %%%
    \begin{cases}
            k\Delta & k\Delta \leq x < (k+1)\Delta \;,\,
            \\
            -k\Delta & -(k+1)\Delta < x \leq -k\Delta \;,\,
            \\
            J\Delta & x\geq J\Delta \;,\,
            \\
            -J\Delta & x\leq -J\Delta \;.\,
            \end{cases}
    \end{align}
    %%%
%%%
Let $\hY = \hX+Z$ be the output corresponding to the input 
$\hX$ and let $\widetilde{Y} = g_{J',\Delta}(\hY)$ be a discretized version of $\hY$ defined in the same manner.
%%%
%\begin{remark}
Observe that the rows of the discretized DMC from $\hX$ to $\widetilde{Y}$ are distinct for sufficiently large $J$ and small $\Delta$, since for every pair of inputs $x_1,x_2\in\mathbb{R}$, $x_1\neq x_2$, we have
% ---
\begin{align}
    f_Z(y-x_1)\neq f_Z(y-x_2) \;,\,    
\end{align}
% ---
for some $y\in\mathbb{R}$ (e.g. $y=x_1$).
%\end{remark}
%%%
Thus, based on Theorem~\ref{Th.DDICapacity}, %we can show that for each finite $J$ and $J'$, 
any rate
% ---
\begin{align}
    R = H(\hX)-\epsilon \;,\,
    \label{Eq.DiscretizRate}
\end{align}
% ---
is achievable for the DMC with input $\hX$ and output $\widetilde{Y}$ under power constraint $A$, where $\epsilon>0$ is arbitrarily small. By  (\ref{Eq.Xhat}), the probability distribution of the discretized variable is specified by
% ---
\begin{align}
    \Pr(\hX=\pm k\Delta) =
    \begin{cases}
        p_k & k \in [\![J-1]\!] \;,\, \\
        \sum_{k=J}^{\infty} p_{k} & k \in \left\{ J,J+1,\cdots \right\} \;,\, \\
        2p_0 & k= 0 \;,\, \\
    \end{cases}
\end{align}
% ---
where
% ---
\begin{align}
\label{Eq.pk}
    p_k = \int_{k\Delta}^{(k+1)\Delta} f_X(x) \, dx \;,\,
\end{align}
% ---
for $k \in \left\{0,\,1,\cdots \right\}$ and for the case $k=0$ we have
% ---
\begin{align}
    \Pr(\hX=0) & = \int_{0}^{\Delta} f_X(x) \, dx
    \nonumber\\
    & = \frac{1}{2} \int_{-\Delta}^\Delta f_X(x) \, dx \;,\,
\end{align}
% ---
% % ---
% \begin{align}
%     \Pr(\hX=\pm k\Delta)&=p_k \,,\; \qquad \text{for $ k \in [\![J-1]\!]$}
%     \nonumber\\
%     \Pr(\hX=\pm J\Delta)&= \sum_{k'=J}^{\infty} p_{k'} 
%     \nonumber\\
%     \Pr(\hX=0)&=\int_{-\Delta}^\Delta f_X(x) \, dx = 2p_0 \;,\,
% \end{align}
% % ---
Thus, the corresponding entropy is bounded by
% ---
\begin{align}
    R + \epsilon
    & = H(\hX)
    \nonumber\\
    & = -\sum_{k=-J}^J \Pr(\hX = k\Delta) \log \Pr(\hX = k\Delta) 
    \nonumber\\
    & \geq
    -\sum_{k=-(J-1)}^{J-1} \Pr(\hX = k\Delta) \log \Pr(\hX = k\Delta)
    \nonumber\\
    \label{Eq.QuantizedEntropy}
    & = -2p_0\log(2p_0)-2\sum_{k=1}^{J-1} p_k \log p_k \;.\,
    \end{align}
% ---
Since the Gaussian density function $f_X$ is continuous, then, by the \emph{mean value theorem}, there exists a value $x_k$ within each discretization interval such that
%%%
\begin{align}
    f_X(x_k) \Delta & = \int_{k\Delta}^{(k+1)\Delta} f_X(x) \, dx
    \nonumber\\
    & = p_k \;,\,
\end{align}
%%%
where the last equality holds by the definition of $p_k$ in (\ref{Eq.pk}). Plugging this into (\ref{Eq.QuantizedEntropy}), we obtain
% ---
\begin{align}
    H(\hX) & \geq
    -2 f_X(x_0)\Delta \log (2f_X(x_0) \Delta)
    -2\sum_{k=1}^{J-1} f_X(x_k) \Delta \log (f_X(x_k)\Delta)
    \nonumber\\
    & = -2 f_X(x_0) \Delta \log (2f_X(x_0))
    -2\sum_{k=1}^{J-1} f_X(x_k)\Delta \log (f_X(x_k))
    \nonumber\\
    & = -2 f_X(x_0)\Delta \log \Delta
    -2\sum_{k=1}^{J-1} f_X(x_k)\Delta \log \Delta \;.\,
\end{align}
% ---
Then, taking $J$ to infinity, we have
%%%
\begin{align}
    \lim_{J\to\infty} H(\hX)&\geq      
    -2 f_X(x_0)\Delta \log (2f_X(x_0))
     -2\sum_{k=1}^{\infty} f_X(x_k)\Delta \log (f_X(x_k))
    -\log \Delta \left( 2\sum_{k=0}^{\infty} f_X(x_k) \Delta \right) 
    \nonumber\\
    &= -2 f_X(x_0)\Delta 
     -2\sum_{k=0}^{\infty} f_X(x_k)\Delta \log (f_X(x_k))
    -\log \Delta \;,\,
    \label{Eq.GentropyDelta1}
\end{align}
%%%
since
%%%
\begin{align}
    2\sum_{k=0}^{\infty}f_X(x_k)\Delta & = \sum_{k=-\infty}^\infty \Pr(X=k\Delta)
    \nonumber\\
    & = 1 \;,\,
\end{align}
%%%
As the Gaussian pdf is bounded by $f_X(x)\leq (2\pi A)^{-1/2}$, the last bound, (\ref{Eq.GentropyDelta1}), implies
%%%
\begin{align}
    \lim_{J\to\infty} H(\hX)&\geq      
     -2\sum_{k=0}^{\infty} \Delta f_X(x_k) \log (f_X(x_k))
    -\frac{2}{\sqrt{2\pi A}}\Delta+\log \frac{1}{\Delta} \;.\,
    \label{Eq.GentropyDelta2}
\end{align}
%%%
At last, we take the limit $\Delta\to 0^+$. 
First, consider the sum.
Since $f_X(x)\log f_X(x)$ is Riemann integrable,
\begin{align}
    \lim_{\Delta\to 0^+}\left(-2\sum_{k=0}^{\infty} \Delta f_X(x_k) \log (f_X(x_k)) \right) & = -2\int_{0}^{\infty} f_X(x)\Delta \log (f_X(x)) \, dx
     \nonumber\\
     & = -\int_{-\infty}^{\infty} f_X(x)\log f_X(x) \, dx \nonumber\\
     & = h(X)
     \nonumber\\
     & = \frac{1}{2}\log(2\pi e A) \;.\,
    \label{Eq.GentropyDelta3}
\end{align}
%%%
The second term in the right hand side of (\ref{Eq.GentropyDelta2}) tends to zero as 
$\delta\to 0^+$.
Hence, as $J\to\infty$ and $\delta\to 0^+$, we obtain $ R+\epsilon = H(\hX)$ converges to
% ---
\begin{align}
    \frac{1}{2}\log(2\pi e A) + \lim_{\Delta\to 0^+}\log \frac{1}{\Delta} \;.\,
\end{align}
% ---
which tends to $\infty$. This completes the proof.
\qed
% ------------------------------
\section{Summary and Discussion}
 \label{Sec.SummaryDiscussions}
To summarize, we have established the deterministic identification (DI) capacity of a channel subject to an input constraint. Our capacity formula is given in terms of the entropy of the channel input. For  a Gaussian Channel $\mathscr{G}$, the DI capacity is  
$\mathsf{C}_{DI}(\mathscr{G})=\infty$, regardless of  the noise (as long as it has finite energy).
Our results have the following geometric interpretation.
At a first glance, it may seem reasonable that for the purpose of identification, one codeword could represent two messages.
While identification allows overlap between decoding regions \cite{KE05,AADT20},
overlap at the encoder is not allowed for deterministic codes.
We observed that when two messages are represented by the same codeword, then, if the  probability of missed identification is upper bounded by $\epsilon$, then the  probability of false identification is lower bounded by $1-\epsilon$. Hence, low  probability of type I error comes at the expense of high probability of type II error, and vice versa.
Thus, deterministic coding imposes the restriction that every message must have a distinct codeword.
The converse proof follows from this property in a  straightforward manner since the volume of the input subspace of sequences that satisfy the input constraint is 
$\approx 2^{n\mathsf{C}_{DI}(\W)}$, with 
% ---
\begin{align}
    \mathsf{C}_{DI}(\W) = \underset{p_X \,:\; \mathbb{E}\{\phi(X)\} \leq A}{\max}~H(X) \;.\,
\end{align}
% --- 
A similar principle guides the direct proof as well. The input space is covered such that each codeword is surrounded by a sphere of radius $\frac{n\epsilon}{2}$ to separate the codewords. For the Gaussian channel, the DI capacity can be achieved using a simple distance-decoder.
 
Next, we compare and discuss different results from the literature on the DI capacity. We will use the notation of $\underline{\mathbb{C}}_{DI}(\W)$ for the DI capacity in the double-exponential scale, or equivalently, when the rate is defined as
% ---
\begin{align}
    R=\frac{1}{n}\log \log (\text{$\#$ of messages}) \;,\,
\end{align}
% ---
as stated in \cite{AD89}, and confirmed in this paper as well,
%%%
 \begin{align}
     \underline{\mathbb{C}}_{DI}(\W)=0 \;,\,
 \end{align}
%%% 
since the code size of DI codes scales only exponentially in block length.

On the other hand, as observed by Bracher and Lapidoth \cite{BL17}, if one considers an average error criterion instead of the maximal error, then the double-exponential performance of randomized-encoder codes can also be achieved using deterministic codes. 

Alternatively, one may consider the $\epsilon$-capacity, for a fixed $0<\epsilon<1$. In the double exponential scale, a rate $R$ is called $\epsilon$-achievable if there exists a
 $(2^{{2^{nR}}},n,\epsilon,\epsilon)$ code for sufficiently large $n$. The DI $\epsilon$-capacity $\underline{\mathbb{C}}^{\epsilon}_{DI}(\W)$ is then defined as the supremum of $\epsilon$-achievable rates.
As the DI and RI capacities in the double exponential scale have strong converses \cite{AD89,HV92,Bur94},
% ---
\begin{align}
    \underline{\mathbb{C}}_{RI}^{\epsilon}(\W) &= \underline{\mathbb{C}}_{RI}(\W)=
    \max_{p_X} I(X;Y) \;,\,
    \\
    \underline{\mathbb{C}}_{DI}^{\epsilon}(\W) &=\underline{\mathbb{C}}_{DI}(\W)= 0 \;,\,
    \intertext{for $0<\epsilon<\frac{1}{2}$. On the other hand, for $\epsilon\geq \frac{1}{2}$ we have}
 \label{Eq.InfiniteCapacity}
    \underline{\mathbb{C}}_{DI}^{\epsilon}(\W) & = \underline{\mathbb{C}}_{RI}^{\epsilon}(\W)= \infty \;.\,
\end{align}
% ---
(see \cite{BV00}). To understand (\ref{Eq.InfiniteCapacity}), suppose  $\epsilon>\frac{1}{2}$, and consider an arbitrary set of codewords with a stochastic decoder that makes a decision for the identification hypothesis by flipping a fair coin \cite{AD89}. Both error probabilities of type I and II equal $\frac{1}{2}$, and are thus smaller than $\epsilon$.
  
 By providing a detailed proof for the DI capacity theorem with and without an input constraint, we have filled the gap in the previous analysis  \cite{AD89,AN99} as well.
 In particular, in \cite{AN99}, Ahlswede and Cai asserted that the DI capacity for a compound channel is given by
 %%%
 \begin{align}
    C_{DI}(\W_{\text{compound}}) = \underset{p_X(x)}{\max}~ \underset{s\in\S}{\min}~H(\hat{X}(s)) \;,\,
 \end{align}
 %%%
 where $s\in\S$ is the channel state, and the map $\hat{X}(s)$ is induced from $X$ by a partition of the input alphabet to equivalent classes as specified in \cite[Sec.~I.F]{AN99}. This result immediately yields Corollary \ref{Co.DDICapacity0}, since the DMC is a special case of a compound channel with a single state value. Indeed, taking $|\S| = 1$ and considering the reduced channel 
$W_r$ (see Definition~\ref{Def.ReducedChannel}),
it can be readily shown that $\hat{X}(s)=X$. Nonetheless, a significant part of the proof in \cite{AN99} is missing. At the beginning of Sec. VII in \cite{AN99}, the following claim is given: ``It was shown in [A'80] that for any channel $\tilde{V}:\X\to¸\Y$ without two identical rows, any $u_1,u_2,\epsilon>0$, sufficiently large $n$ and any $\U\subset\X^n$ such that for all $u,u'\in\U, d_H(u,u')>n\epsilon$, there exists a family of subsets of $\Y^n$, say $\D_u,u\in\U$, such that $\tilde{V}^n(\D_u|u)>1-u_1$ and $\tilde{V}^n(\D_u|u')<u_2$ for all $u\neq u'$, where $d_H$ is the Hamming distance.", where [A'80]  refers to a paper by Ahlswede \cite{A80} on the arbitrarily varying channel, and does not include identification. In fact, the work \cite{A80} was published 9 years before the introduction of the identification problem by Ahlswede and Dueck \cite{AD89}. 
 Unfortunately, a straightforward extension of the methods in \cite{A80}, using decoding territories, does not seem to yield a proof of such a property \cite{CaiEmail2}. In this sense, our derivation completes the proof of Ahlswede and Cai's capacity theorem \cite{AN99} for the compound channel.
%%%%%%%%%%%%%%%%%%%%%%%%%%%%%%%%%%%%%%%%%%%%%%%%%%%%%%%%%%%
\section*{Acknowledgments}
We gratefully thank Ning Cai (ShanghaiTech University) for a discussion on the DI capacity for DMCs. We thank Robert Schober (Friedrich Alexander University) for discussions and questions about the application of identification theory in molecular communications.

Mohammad J. Salariseddigh, Uzi Pereg, and Christian Deppe  were supported by 16KIS1005 (LNT, NEWCOM). Holger Boche was supported by 16KIS1003K (LTI, NEWCOM) and the BMBF within the national initiative for “Molecular Communications (MAMOKO)” under grant 16KIS0914.
% Mohammad J. Salariseddigh, Uzi Pereg, and Christian Deppe  were supported by 16KIS1005 (LNT, NEWCOM). Holger Boche was supported by 16KIS1003K (LTI, NEWCOM).
%%%%%%%%%%%%%%%%%%%%%%%%%%%%%%%%%%%%%%%%%%%%%%%%%%%%%%%%%%%
\appendices
%%%%
\section{Proof of Lemma~\ref{Lem.Reduction}}
\label{App.Reduction}
%%%
Let $\W$ be a given DMC, with a stochastic matrix $W:\X\to\Y$ and its reduced version $W_r:\X_r\to\Y$ as defined in Definition~\ref{Def.ReducedChannel}. 
Observe that the capacity of the original channel is lower bounded by that of the reduced channel, i.e.,
%%%
    \begin{align}
       \mathbb{C}_{DI}(\W) \geq \mathbb{C}_{DI}(\W_r) \;,\,
    \end{align}
%%%
since every code for $W_r$ can also be used for $W$. Hence, it remains to be shown that 
    \begin{align}
       \mathbb{C}_{DI}(\W_r) \geq \mathbb{C}_{DI}(\W) \;.\,
    \end{align}
%%%

Assume without loss of generality that the input alphabet of the original channel $\W$ is given by $\X = \{1,2,\cdots,|\X|\}$. 
Let $L~:~\X\to\X_r$ denote the projection of the input alphabet onto the equivalent classes, 
%%%
\begin{align}
    &L[x] = z(\ell) \quad \text{iff} \quad x\in\X(\ell) \;.\,
\end{align}
%%%
Now let $(\U,\D)$ be a $(2^{nR},n,\lambda_1,\lambda_2)$ code for $\W$. Then the type I probability of error can be expressed as
%%%
\begin{align}
    P_{e,1}^W(i) & = \sum_{y^n\notin \D_i} W^n(y^n|u_i)
    \nonumber\\
    & = \sum_{y^n\notin \D_i} \prod_{t=1}^n W(y(t)|u_{i}(t)) \;,\,
\end{align}
%%%
where we use the notation $y^n = \big(y(t)\big)_{t=1}^n$ and  $u_i = \big(u_i(t)\big)_{t=1}^n$.
%%%
Next, we define a code $(\widetilde{\U},\D)$ for the channel $\W_r$ where the codebook consists of the following codewords,
%%%
\begin{align}
    &\tilde{u}_i = \big(L[u_{i}(t)]\big)_{t=1}^n \;.\,
\end{align}
%%%
Now recall that we have defined the equivalence classes such that input letters in the same equivalence class correspond to identical rows in the channel matrix $W$ (see Definition~\ref{Def.ReducedChannel}). Thus, by definition, 
%%%
\begin{align}
    W_r(y|L[x]) = W(y|x) \;,\,
\label{Eq.Lchannel}
\end{align}
%%%
for all $x\in\X$ and $y\in\Y$. Hence, the error probability of type I for the reduced channel $\W_r$ satisfies
%%%
\begin{align}
    P_{e,1}^{W_r}(i) &\stackrel{(a)}{=} \sum_{y^n\notin \D_i} W_r(y^n|\tilde{u}_i) \nonumber\\
    &\stackrel{(b)}{=}\sum_{y^n\notin \D_i} \prod_{t=1}^n W_r\big(y(t)|L[u_{i}(t)] \big)\nonumber\\
    &\stackrel{(c)}{=} \sum_{y^n\notin \D_i} \prod_{t=1}^n W(y(t)|u_{i}(t)) \nonumber\\
    &\stackrel{(d)}{=} \sum_{y^n\notin \D_i} W ^n(y^n|u_i)\nonumber\\
    &\stackrel{(e)}{=} P_{e,1}^{W}(i) \;,\,
\end{align}
%%%
for all $i$, where $(a)$ and $(e)$ are due to (\ref{Eq.TypeIErrorDef}); $(b)$ and $(d)$ hold since the channel is memoryless, and $(c)$ follows from (\ref{Eq.Lchannel}). By the same considerations, we also have $P_{e,2}^{W_r}(i,j)=P_{e,2}^{W}(i,j)$ for all $ j\neq i$. That is, the error probabilities of the code 
$(\widetilde{\U},\D)$ are the same as those of the original code for $\W$.
Therefore, the code constructed above for $\W_r$ is also  a $(2^{nR},n,\lambda_1,\lambda_2)$ code, and the proof of Lemma~\ref{Lem.Reduction} follows.
\qed
%%%
%%%%%%%%%%%%%%%%%%%%%%%%%%%%%%%%%%%%%%%%%%%%%%%%%%%%%%%%%%%

\bibliographystyle{IEEEtran}
\bibliography{IEEEabrv,confs-jrnls,refrences}
%\printbibliography

\end{document}